\documentclass[hidelinks]{article}

\usepackage{arxiv}

\usepackage{natbib} 
\usepackage[utf8]{inputenc} 
\usepackage[T1]{fontenc}    
\usepackage{hyperref}       
\usepackage{url}            
\usepackage{booktabs}       
\usepackage{amsfonts}       
\usepackage{nicefrac}       
\usepackage{microtype}      
\usepackage{xcolor}         

\usepackage{graphicx,subcaption}
\usepackage{amsmath,amsthm}
\usepackage{multirow}
\usepackage{changepage}
\newtheorem{lemma}{Lemma}
\usepackage{algpseudocode}
\usepackage{algorithm2e}
\RestyleAlgo{ruled}

\newcommand{\modelname}{HG2NP}

\title{Heterogeneous Graph Generation: A Hierarchical Approach using Node Feature Pooling}

\author{%
  Hritaban~Ghosh \\
  Indian Institute of Technology Kharagpur, India\\
  \texttt{ghoshhritaban21@iitkgp.ac.in} \\
  \And
  Chen~Changyu \\
  Singapore Management University, Singapore \\
  \texttt{cychen.2020@phdcs.smu.edu.sg} \\
  \AND
  Arunesh~Sinha \\
  Rutgers University, Newark, USA\\
  \texttt{arunesh.sinha@rutgers.edu } \\
  \And
  Shamik~Sural \\
  Indian Institute of Technology Kharagpur, India\\
  \texttt{shamik@cse.ac.in} \\
}

\date{}

\renewcommand{\headeright}{}
\renewcommand{\undertitle}{}
\renewcommand{\shorttitle}{Heterogeneous Graph Generation: A Hierarchical Approach}
\begin{document}

\maketitle

\begin{abstract}
  Heterogeneous graphs are present in various domains, such as social networks, recommendation systems, and biological networks. Unlike homogeneous graphs, heterogeneous graphs consist of multiple types of nodes and edges, each representing different entities and relationships. Generating realistic heterogeneous graphs that capture the complex interactions among diverse entities is a difficult task due to several reasons. The generator has to model both the node type distribution along with the feature distribution for each node type. In this paper, we look into solving challenges in heterogeneous graph generation, by employing a two phase hierarchical structure, wherein the first phase creates a  skeleton graph with node types using a prior diffusion based model and in the second phase, we use an encoder and a sampler structure as generator to assign node type specific features to the nodes. 
  A discriminator is used to guide training of the generator and feature vectors are sampled from a node feature pool. We conduct extensive experiments with subsets of IMDB and DBLP datasets to show the effectiveness of our method and also the need for various architecture components.
\end{abstract}

\section{Introduction}
Heterogeneous graphs appear ubiquitously in various domains, including social networks, recommendation systems, and biological networks. Heterogeneous graphs, as opposed to homogeneous ones, have multiple node types and sometimes multiple edge types. They allow for the modeling of diverse entities and their interconnections, providing a richer and more comprehensive view of the underlying system, such as in citation networks~\cite{zhou2007co}, social networks ~\citep{dong2012link}, and biological networks ~\citep{chen2012drug,li2010genome,wang2014drug}. In a social network, for example, an individual can be part of a group and groups can be part of a community. Even in this simplified description of a social network, several node types can be identified. For instance, node types in this network are  ``individual'', ``group'', and ``community''. In this work, we focus on generating such heterogeneous graphs that in turn can improve the accuracy and reliability of predictive analytics, anomaly detection, and risk assessment systems built from such synthetic data. This further can lead to improved performance in various tasks such as recommendation~\citep{chen2021graph}, classification~\citep{santos2018representation}, and link prediction~\citep{dong2012link}.

The level of detail and meaningful relationships captured by heterogeneous graphs cannot be achieved by homogeneous graphs without several assumptions and constraints. For example, if one were to model node types in a social network using a homogeneous graph, one would have to set aside components in the node feature vector to indicate the discrete node type, which presents an issue when the node features are real valued as we mix real valued and discrete valued features in one vector. To add to the complexity, node features can depend on the type, thus, for example, if a feature vector for individual type node can be encoded in $i$-bits, a feature vector for group type can be encoded in $j$-bits and a feature vector for community type can be encoded in $k$-bits, the node feature vector in the homogeneous graph must be at least $\max(i, j, k)$-bits and zeroes would fill up unused bits for shorter feature vectors. This leads to a wastage of resources and therefore, it is imperative to use more advanced data structures such as heterogeneous graphs.

In recent years, there has been a tremendous effort in the field of homogeneous graph generation \citep{guo2022systematic,you2018graphrnn,su2019graph,d2019deep,li2018learning,simonovsky2018graphvae,vignac2022digress}. In the area of graph neural networks~\citep{hu2020heterogeneous}, heterogeneous graphs such as IMDB and DBLP have been investigated. However, to the best of our knowledge, there is only one prior work in heterogeneous graph generation~\citep{ling2021deep}, which is lacking in many aspects (see related work). In this work, we introduce a novel hierarchical approach called Heterogeneous Graph Generation with Node Pools (\modelname) for heterogeneous graphs which stands on the shoulders of work on homogeneous graph generation. 

\textbf{Problem Statement}: We denote a heterogeneous graph by $G = (V, E)$. 
Nodes have types and let $K$ be the number of node types. 
Given a set of $M$ observed heterogeneous graphs $\{G_i\}_{i=1}^M$, the heterogeneous graph generation problem is to learn
the probability distribution of these 
graphs $\mathbb{P}$, from which new graphs can be sampled $G \sim \mathbb{P}$.

Our \emph{first contribution} is to divide the problem into two easier tasks, accomplished in two phases. In the first phase, we generate a skeleton graph with node types while in the second phase, we assign features to the nodes in the skeleton graph. As a \emph{second contribution}, we make use of message passing (MP) and a node feature vector pool to design an adversarial generator-discriminator architecture for the second phase. The MP module provides neighborhood information that is critical for generating features of nodes. The node features are sampled from a pool of feature vectors constructed from all the data, as we find that generation of real valued large dimensional vectors with no restriction produces poor quality outcomes.
Also, while the problem is similar to any other generative problem, measuring the fidelity of generated heterogeneous graphs requires new domain specific metrics that take nodes types into account. As our \emph{third contribution}, we introduce two new metrics for such measurement. Our \emph{fourth contribution} is the concept of labeled permutation equivariance and invariance, which are relevant for graphs with node types, and a proof showing that our model possesses these properties. Finally, we perform comprehensive experiments with IMDB and DBLP graph data to show the superiority of our method compared to baselines.


\section{Related Work and Background}

Graph generation holds immense practical value, with applications ranging from the design of molecular structures in drug development~\citep{gomez2018automatic,li2018learning,li2018multi,you2018graph} to the optimization of model architectures through generative models for computation graphs~\citep{wasserman1996logit}. In \cite{zhu2022survey}, the authors survey deep graph generation techniques and their practical applications. In~\cite{guo2022systematic}, the authors highlight the formal definitions and taxonomies of several deep generative models for homogeneous graph generation including but not limited to \cite{you2018graphrnn}, \cite{su2019graph}, \cite{d2019deep}, \cite{khodayar2019deep}, \cite{zhang2019d}, \cite{assouel2018defactor}, \cite{lim2020scaffold}, \cite{li2018learning}, \cite{liu2018constrained}, \cite{kearnes2019decoding}, \cite{bresson2019two}, \cite{guarino2017dipol}, \cite{flam2020graph}, and \cite{niu2020permutation}. \cite{guo2022systematic} also provide an overview of the evaluation metrics in this specific domain. In \cite{kipf2016variational}, the authors introduce the variational graph auto-encoder (VGAE), a framework for unsupervised learning on graph-structured data based on the variational auto-encoder (VAE) as originally introduced in \cite{simonovsky2018graphvae}. In recent years, alongside the aforementioned approaches, diffusion-based methods have gained prominence, with several studies conducted in this area~\citep{vignac2022digress,qin2023sparse,jo2023graph}. The DiGress model~\citep{vignac2022digress} utilizes a discrete diffusion process that progressively edits graphs with noise, through the process of adding or removing edges and changing the categories. A graph transformer network is trained to revert this process. This model is used as the skeleton graph generator in the first phase of our approach.

As opposed to the vast majority of graph generation papers listed above and in surveys,~\citet{bojchevski2018netgan} proposed a generative approach called NetGAN that was trained only on one input graph. This different paradigm aims to generate graphs that match properties of the single graph used for training, thus, there is no learning of distribution over graphs. 
The only work that we know of for heterogeneous graph generation~\citep{ling2021deep} is also based on the single graph training paradigm. They present a heterogeneous path generator designed to path instances. The paths are a sequence of node types and edge types, effectively encapsulating  semantic information between their endpoints. Additionally, they propose a heterogeneous graph assembler capable of  stitching the generated path instances into one heterogeneous graph in a stratified manner. 
Expanding on their own work,~\citet{ling2023motif} extend the methodology by incorporating a graph motif generator. This aims to enhance the characterization of higher-order structural distributions within heterogeneous graphs. The significantly different framework and metrics for this work makes it inapplicable to our problem of learning the distribution of graphs and then generating from the learned distribution. Additionally, as the code from this paper is not reproducible for new datasets, we were unable to even adapt the work for comparison.

\textbf{Message Passing Background}: Message passing in a graph allows to effectively incorporate connectivity information between nodes. Assume a graph $G = (V,E)$ with initial node features $x_u$ for node $u \in V$ and edge features $e_{uv} \in E$ (edge feature might be empty). Let $N(u)$ denote neighbor nodes of $u$. Message passing is defined by an operation
$h_u = U(x_u, \bigoplus_{v \in N(u)} M(x_u, x_v, e_{uv}))$ that updates the current node features $x_u$ to $h_u$ and is done repeatedly for all nodes. Here, $U$ is an update function, $\bigoplus$ is an aggregation function and $M$ is a message function (builds message from $v$ to $u$). There are many variations of these functions in the literature. We use a specific popular one called graph attention networks~\citep{brody2021attentive} which builds these functions based on attention mechanism; see the reference for details.

\textbf{DiGress Background}: 
DiGress is a discrete denoising diffusion model for generating homogeneous graphs with categorical node and edge attributes. We briefly describe DiGress~\citep{vignac2022digress}. There are three types of discrete noises in the diffusion model, namely, categorical edits in which the type of a node is changed, edge additions in which a new edge is added, and edge removals in which an existing edge is removed. A graph transformer network is trained to revert this process.  For any node, the transition probabilities are defined by the matrices $[Q^{t}_{X}]_{ij} = q(x_{t} = j | x^{t-1} = i)$ and similarly for edges, they are defined as $[Q^{t}_{E}]_{ij} = q(e_{t} = j | e^{t-1} = i)$. Furthermore, adding noise to current graph to form $G^{t} = (X^{t}, E^{t})$, implies sampling each node type and edge types from a categorical distribution given by $q(G^{t}|G^{t-1}) = (X^{t-1}Q^{t}_{X}, E^{t-1}Q^{t}_{E})$. Extending this over $t$ timesteps, the following is obtained:
$q(G^{t}|G) = (X^{t-1}\bar{Q}^{t}_{X}, E^{t-1}\bar{Q}^{t}_{E})$
where $\bar{Q}^{t}_{X} = Q^{1}_{X}...Q^{t}_{X}$ and $\bar{Q}^{t}_{E} = Q^{1}_{E}...Q^{t}_{E}$.
Lastly, the denoising neural network $\phi_{\theta}$ takes a noisy graph $G^{t} = (X^{t}, E^{t})$ as input and predicts a clean graph $G$. Once the network is trained, it can be used to sample new graphs. We use a fine-tuned DiGress model on our datasets to obtain skeleton graphs for our second phase.

\section{Methodology} \label{sec:methodology}
In this section, we describe our model \modelname\ in detail. First, we note that as per our problem statement in the introduction, we aim to generate node types but not explicit edge types. This does not mean that we cannot express edge types, since we can implicitly define an edge type by the node types of its endpoints in many scenarios.
In fact, this suffices for the edge relationships found in the standard DBLP and IMDB data available from PyTorch Geometric~\citep{Fey/Lenssen/2019}. 
Nonetheless, in case of richer edge labels that cannot be defined by the type of endpoints, one can always convert such a graph to one with node types only by introducing a node $n_e$ for each labeled edge $e$ (edge $e = (v,u)$ is replaced by unlabeled edges $(v,n_e)$ and $(n_e, u)$) and assigning the edge label of $e$ to this new node $n_e$, albeit at the cost of a larger graph.

\modelname\ is a two phase scheme, in which the first phase leverages a state-of-the-art existing homogeneous graph generation framework to output a \emph{skeleton graph} with nodes and edges, and node type. Then, in the second phase, we assign feature vectors to the nodes. The second phase is set-up as a generative adversarial network. An overview is shown in Figure~\ref{fig:overall}. Details of each of the phases are provided in the next sub-sections.

\begin{figure}[t]
    \centering
    \fbox{\includegraphics[width=0.99\textwidth]{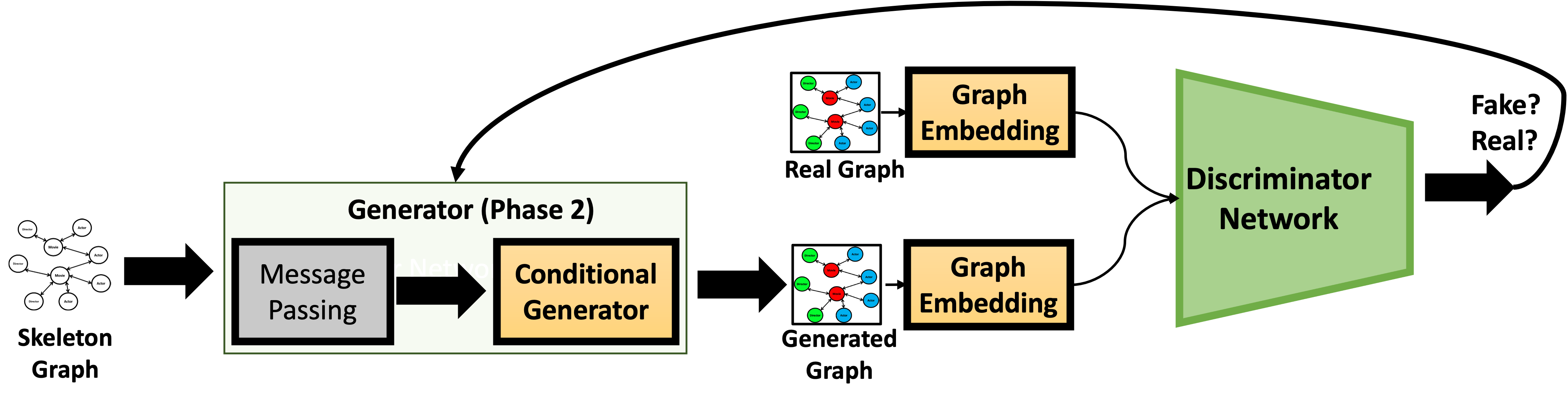}}
    \caption{Overall generation process for \modelname. The skeleton graph is the output of Phase 1. Phase 2 uses a GAN type structure (see details in text and Figure~\ref{fig:generation}).}     \label{fig:overall}
\end{figure}

\subsection{Phase 1: Skeleton Graph Generation}

We leverage existing state-of-the-art model DiGress~\citep{vignac2022digress} designed for generating homogeneous graphs. DiGress is a diffusion based generative model; we utilize this to learn the distribution patterns of node type underlying the heterogeneous graphs in our datasets. For training, we extract real skeleton graphs from the dataset for training of the generator by keeping only the node, node type and edge information. 
By breaking down the complex task of heterogeneous graph generation into two distinct phases, we intend to employ effective models for the easier tasks in each phase.

\subsection{Phase 2: Heterogeneous Feature Vector Assignment}

Our heterogeneous node feature vector assignment step is based on the adversarial generator-discriminator architecture, conditional on the skeleton graph generated in Phase 1.

\textbf{Multi-type Multi-generator}: Given a skeleton graph $S$ with node types $x_u$ from Phase 1, in the second phase, we first utilize a message passing module to update each node's type vector to obtain a graph $U(S)$ with updated node type vectors $h_u$ for each node $u$. We utilize graph attention based message passing developed in prior work~\citep{brody2021attentive}. Based on the intuition of message passing, the updated node type vector captures the neighborhood node type information. We also keep a copy of the original node type vector to be used later as described next.  There is a separate conditional generator $g_k(x_u, h_u; \theta_k)$ for each node type $k$. From the original node type vector, we choose the corresponding $g_k$ and then, conditioned on the updated node type vector for this node, we sample a feature vector from a node type specific pool of node feature vectors. 
The conditional generators are implemented using a simple fully connected neural network with weights $\theta_k$ with Leaky ReLU activation and a softmax at the end followed by sampling of feature. Note that the softmax is over the node pool of feature vectors of each type, which is why the network is different for each node type. In this way, all the nodes of the graph are assigned feature vectors. This yields a generated heterogeneous graph. The full generation (including Phase 1) is shown by an example in Figure~\ref{fig:generation}. The intuition behind using a pool of node feature vectors is based on our preliminary experiments where we found that learning to generate node feature vectors trained very slowly and always resulted in poor quality node features; we verify this with an ablation in experiments.

\begin{figure}[t]
    \centering
    \fbox{\includegraphics[width=0.99\textwidth]{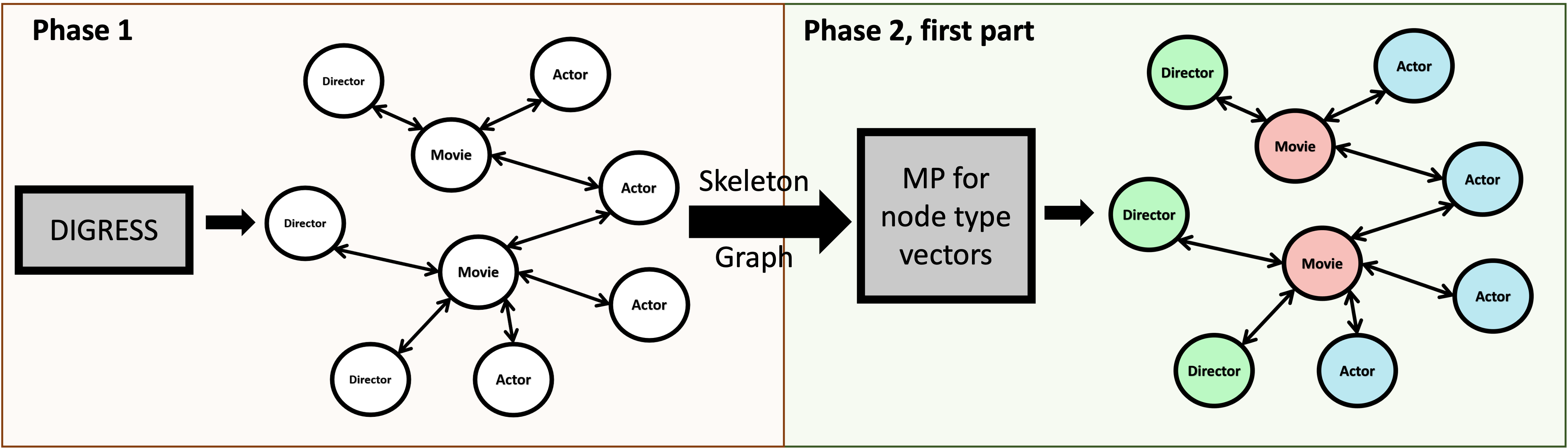}}
    \fbox{\includegraphics[width=0.99\textwidth]{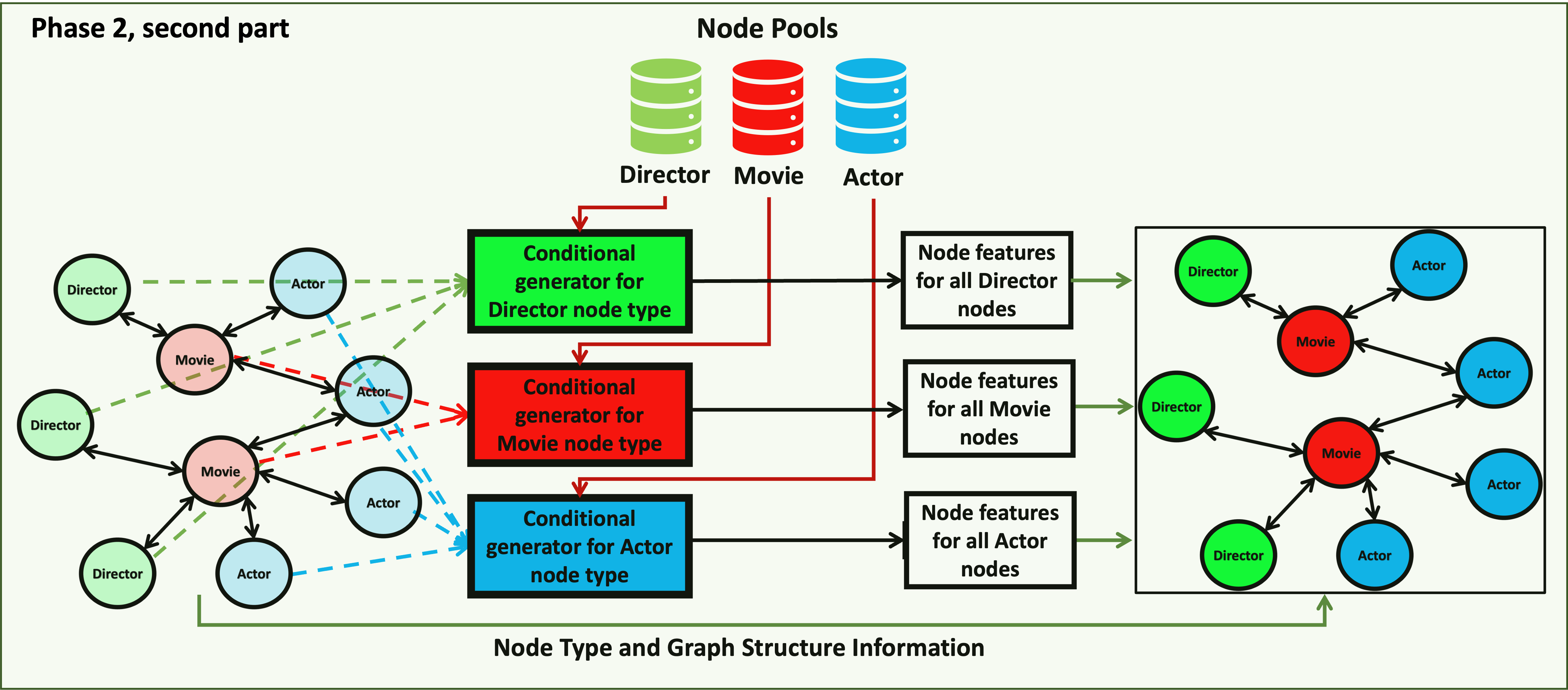}}
    \caption{The generation process using IMDB data as example. In Phase 1, we use DiGress to output a graph with node types, after which the first part in Phase 2 is to perform message passing to embed neighbor node type information in each node type vector.}     \label{fig:generation}
\end{figure}

Formally, we let $\theta = \{\theta_k\}_{k=1}^K$ be the weights of all the generators and the final generated graph be $g_{\theta}(U(S))$, where $g_{\theta}$ is a function that takes $U(S)$ and runs $g_k(x_u, h_u; \theta_k)$ for each node $u$ of type $k$ to get the generated heterogeneous graph. Note that we will need to differentiate $g_\theta$, which in turn requires differentiating through $g_k(\cdot, \cdot; \theta_k)$ and $U$. The MP update $U$ is differentiable by construction, but $g_k(\cdot, \cdot; \theta_k)$ is not since it involves a sampling step from the softmax output. To handle this, we use the well known Gumbel softmax trick from the literature, which approximates non-differentiable argmax by the differentiable softmax in the backward gradient pass to enable passing back gradients through the sampling operation. Thus, we will consider $g_{\theta}(U(S))$ to be differentiable with respect to $\theta$.

\textbf{Discriminator}: The generated heterogeneous graphs and the real heterogeneous graphs from the dataset are used to train a discriminator. The discriminator also employs an MP module to obtain an embedding for the heterogeneous graph. In particular, we use a graph attention based MP module that is specifically designed to work with heterogeneous graphs~\citep{hu2020heterogeneous}. The MP module itself reduces the dimension of node vectors of different node types to a single latent dimension. Then, the node vectors of each type are averaged independently and the averaged vectors are concatenated in a given order of types to form an embedding of the full heterogeneous graph. Call this averaging and concatenating process as function $F$. This latent vector is passed through a simple fully connected neural network with Leaky ReLU activation and a softmax at the end that outputs the probability of the graph being real. 
Formally, given a graph $G$, the heterogeneous MP module produces an updated graph $U_h (G)$ which is processed to a single vector $F(U_h (G))$ that forms the input to discriminator $D_w$ producing $D_w(F(U_h (G))) \in [0,1]$ as the prediction for fake (0) or real (1). Note that $U_h$ and $F$ are differentiable by construction.

\textbf{Loss Functions}: For the discriminator, since the task is a standard classification task, the loss used is binary cross entropy loss.
Also, the generator loss to minimize with respect to $\theta$ is given by $\log(1 - D_w(F(U_h (g_{\theta}(U(S)))))$; recall that we have used the Gumbel softmax trick~\citep{jang2016categorical} to make $g_\theta$ differentiable.

\subsection{Evaluation Metrics}
\label{Evaluation Metrics}
We define two new evaluation metrics specific for heterogeneous graphs.

\textbf{Feature-level EMD distance}: First, we measure how different two given heterogeneous graphs $G_r$ and $G_g$ are, denoted by $D(G_r, G_g)$. Denote the empirical feature distribution conditioned on the node type $i$ as $P_r(f|i)$ for $G_r$ and $P_g(f|i)$ for $G_g$. In order to compute EMD (or Wasserstein distance in continuous space), we need a distance (or cost) between feature vectors. We have feature vectors that are bag of words or GloVe vectors~\cite{pennington2014glove}. For bag of words model, the use of Jaccard distance and for 
GloVe vectors, the use of Euclidean distance are well accepted distances used in literature~\citep{calvo2020evaluation,singh2021text}. We compute the Wasserstein distance $W(P_r(f|i), P_g(f|i))$ between the $P_r(f|i)$ and $P_g(f|i)$ using the feature vectors for node type $i$ s in each graph and invoking a standard optimal transport solver~\citep{flamary2021pot}. Next, we note that there are two empirical node type distributions $T_r$ and $T_g$ in both the graphs.
We  define
$$
D(G_r, G_g) = \frac{\mathbb{E}_{i \sim T_r} [D(i)] + \mathbb{E}_{i \sim T_g} [D(i)]}{2} \;\mbox{ where }\; D(i) = \mbox{W}(P_r(f|i), P_g(f|i))
$$

Next, given a set of sampled real and generated graphs from some real distribution $\mathbb{P}$ and fake distribution $\mathbb{P}_g$, we use the above $D$ to compute pairwise distances between all pairs of graphs. Then, we estimate the Feature EMD as the EMD computed on these sampled graphs with distance $D$. A detailed algorithm for the computation is in the appendix.

\textbf{Type Degree Distribution Metric}: In homogeneous graph generation, a popular metric used is degree distribution MMD. This metric is computed between two sets of graphs, a real set and a generated set, by first determining the degree distribution (a histogram of degrees normalized to 1) for every graph in two sets of graphs. Then, the Maximum Mean Discrepancy (MMD) is computed using these degree distributions; see~\citep{o2021evaluation} for details. We modify this to work with node types. Essentially, for every node, we define a type $i$ degree as the number of nodes of type $i$ that this node is connected to. Next, we obtain a type degree distribution for each of the $K$ types, followed by a type degree distribution MMD (given graph samples) for each type. finally, we average the obtained $K$ type degree distribution MMDs (i.e., add and divide by $K$) to get the type degree distribution metric. 

\subsection{Permutation Equivariance and Invariance}
In graph processing, neural network architectures that produce outputs that are equivariant (in case of output per node) or invariant, to isomorphic input graphs are desired, as they remove the need for data augmentation and generally make training simpler. Isomorphism is defined by edge preserving permutations of nodes~\citep{hsieh2006efficient}, hence the neural networks that satisfy these properties are called permutation equivariant and permutation invariant. In our work, there are three parts of the architecture: Phase 1 of generator, Phase 2 of generator and the discriminator in Phase 2. Phase 1 generator is just DiGress, which has been shown to be permutation equivariant and permutation invariant~\citep{vignac2022digress}. 

For the networks in Phase 2, which handles heterogeneous graphs with node types, we need to appeal to a notion of isomorphism for labeled graphs~\citep{hsieh2006efficient}. Here, we treat the type of node as the label for the node. Briefly, two labeled graphs are isomorphic iff there exists a bijection $f$ between nodes such that it is (1) edge preserving: if $(u,v)$ is an edge iff $(f(u),f(v))$ is an edge and (2) label preserving: $u$ and $f(u)$ have the same label. Hence, we use the term \emph{labeled permutation equivariance or invariance} to denote output equivariance or invariance to isomorphic labeled input graphs. Note that any network that is permutation equivariance or invariance is also labeled permutation equivariance or invariance, as labeled permutation is stricter than permutation because labeled permutation requires preserving labels in addition to preserving edges.

\begin{lemma}
Phase 2 generator is labeled permutation equivariant. 
\end{lemma}
\begin{proof}
    Consider two graphs, one of which is labeled permuted using function $f$. The graph attention message passing is permutation equivariant, hence labeled permutation equivariance. Which means that the node features $x_u$ after Phase 2, part 1 is the same as $x_{f(u)}$ for permuted node $f(u)$. Also, the label for $u$ and $f(u)$ are same, say $l$, since we consider labeled permutation. Thus, the input to the conditional generator of type $l$ is same $x_u = x_{f(u)}$. Then, the probability distribution over node features for node $u$ and node $f(u)$ in permuted graph will also be the same. Thus, the Phase 2 generator is {\color{black} labeled permutation equivariant}.
\end{proof}
\begin{lemma}
Phase 2 discriminator is labeled permutation invariant.
\end{lemma}
\begin{proof}
    Consider two graphs, one of which is labeled permuted using function $f$. The heterogeneous graph attention (HGT)~\cite{hu2020heterogeneous} message passing has operations that are specific to the node type (of the nodes involved). {\color{black} Here, we do not consider edge types since our problem has only node types, thus we treat all edge types in~\cite{hu2020heterogeneous} as just one type (i.e., hence, edge type specific weight matrices are all the same matrix)}.  The labeled permutation also maps the label as is, thus, all the attention weights computed and the messages computed are same for a target node $u$ and $f(u)$ in the permuted graph. The aggregate operation in HGT is just a sum which is permutation invariant. Thus, the updated node vector after one round of HGT is same for all $u$ and $f(u)$ in permuted graph. {\color{black} Hence, HGT is permutation equivariant when graphs have node types only}. Then, after HGT, our operation of averaging node vectors over types (labels) and concatenating them (in a fixed order of types) is invariant to isomorphic labeled graphs. Given this, the input to the classifier discriminator neural network is same for two labeled permuted graphs and hence the output will also be same.
\end{proof}

\section{Experiments}\label{sec:experiments}

\textbf{Data}:
In this work, we use DBLP and IMDB datasets modified from~\cite{fu2020magnn} and~\cite{Fey/Lenssen/2019}. These datasets can be accessed from the Pytorch Geometric package. The whole DBLP dataset is a heterogeneous graph consisting of authors (4,057 nodes), papers (14,328 nodes), terms (7,723 nodes), and conferences (20 nodes). The authors are divided into four research areas (database, data mining, artificial intelligence, information retrieval). Each author node is described by a bag-of-words representation~\citep{zhang2010understanding} of their paper keywords - a vector of size 334. Each paper is described using a bag-of-words representation of the title - a vector of size 4231. Each term is described using GloVe vectors~\citep{pennington2014glove}, which is a vector of size 50. 

The DBLP graph is very large and therefore, we have used some categorical features sourced from the original DBLP data from the DBLP website, a snapshot release of February 2024, to split the large graph into smaller components each corresponding to a unique category of the features. For this, we extract a categorical feature which is the type of the publication. With our available hardware, we are able to process graphs up to 200 nodes.
For the DBLP dataset, the split using the values of author research area (indicated as author) and values of conference, yielded 25 graphs with less than 200 nodes. Additionally, the split using the values of author research area (indicated as author), values of conference, and values of publication type, yielded 77 graphs with less than 200 nodes. These are the two datasets of DBLP based graphs we use in our experiments.

Similarly, the IMDB dataset is also modified. The whole IMDB dataset is a heterogeneous graph consisting of three types of nodes: movies (4,278 nodes), actors (5,257 nodes), and directors (2,081 nodes). The movies are divided into three classes (action, comedy, drama) according to their genre. Movie features have a bag-of-words representation of its plot keywords - a vector of size 3066. Features are assigned to directors and actors as the means of their associated movies' features, which are vectors of size 3066. 
We extract three categorical features, namely, year, language and country. Then, we used the movie classes (indicated as movie), and the three newly extracted categorical features to split the larger graph into a set of smaller graphs. After the processing stage, the number of graphs with less than 200 nodes split by year is 64, by country, year is 441, by country, language, movie is 179 and by country, language, movie, year is 802.

\textbf{Experiment Set-up}: 
All our experiments were performed on a GPU with 16 GB RAM, except the first phase of our approach using DiGress, which is particularly memory intensive, and required GPUs with 32 GB of RAM. The CUDA version of our GPUs was 12.2. All the experiments were performed with the following split for every dataset: 72\% (90\% of 80\%) train, 8\% (10\% of 80\%) validation and 20\% test. The two phases of our experiments are trained independently and therefore the hyperparameter decisions for each phase are also made independently. For the first phase, variations of the learning rate, and optimizer of the DiGress model were tried and the best settings were chosen after a grid search for each of the datasets. This best setting was used to produce the skeleton graphs for the second phase. In a similar vein, for the second phase, variations of the learning rate, and optimizer were tried along with the overall number of parameters in the model. 
Post grid search the best setting was chosen based on the evaluation metrics to produce the final graphs. The hyperparameters are in the appendix.


To evaluate our model and compare with the baselines, we use the new metrics proposed for heterogeneous graphs in Section~\ref{Evaluation Metrics} (called Type Degree and Feat. EMD in tables).  We also use a number of standard metrics used in the literature. All our metrics are computed on samples of real and generated graphs; we generate 10 such sample sets and report the mean and standard deviation in our results. Note that all these metrics in the literature are Maximum Mean Discrepancy (MMD) distance computed for the quantities listed below. The MMD is estimated using these quantities computed for each of the sampled generated and real graphs with 500 sampled graphs. These quantities are 
(1) Degree Distribution (called Degree Dist. in tables), (2) Clustering Coefficient (called Clust. Coeff. in tables), (3) Spectral Density (called Spect. Dens. in tables), (4) Motif Distribution (called Motif Dist. in tables), and (5) Orbit Distribution (called Orbit Dist. in tables). In general, all of these metrics capture the structural similarity of the generated graphs to that of the real graphs in the dataset. All the metrics (including new ones) are better when lower.

\textbf{Baselines}: As stated earlier, there is only one prior work~\citep{ling2021deep} on heterogeneous graph generation and as described in related work, this work is not applicable to training over multiple graphs. Hence, we make two baselines by modifying the prior work VGAE~\citep{kipf2016variational} for our scenario. VGAE is a Variational Autoencoder based homogeneous graph generative model. They make use of homogeneous MP to compute intermediate latent vectors which are then used to construct a new adjacency matrix.
We trained VGAE to produce skeleton graphs for our datasets. Then, in order for VGAE to work with heterogeneous graphs, we added a random feature vector assignment from the node pool of the appropriate node type as a post-processing module to construct a heterogeneous graph. It is important to note that node types in VGAE are copied from the input and are not manipulated in its generation process (only edges are manipulated), providing it with significant advantage as a baseline. 

 \begin{table}[t]
 \centering
\setlength{\tabcolsep}{2.5pt}
\caption{IMDB results for four splits (bold shows best results).} \label{tab:IMDBmain}
\begin{tabular}{l c c c c c c}
\toprule
 & \multicolumn{3}{c}{Year} & \multicolumn{3}{c}{Country and Year} \\
 \cmidrule(lr){2-4}
\cmidrule(lr){5-7}
Metrics & VGAE & VGAE-H & \modelname & VGAE & VGAE-H & \modelname\\
\midrule 
Degree Dist. 
& 0.35 $\!\pm\!$ 0.004 
& 0.34 $\!\pm\!$ 0.008 
& \textbf{0.18 $\!\pm\!$ 0.001}

& 0.29 $\!\pm\!$ 0.003 
& 0.29 $\!\pm\!$ 0.003 
& \textbf{0.15 $\!\pm\!$ 4.8e-4} \\

Clust. Coeff. 
& 0.92 $\!\pm\!$ 0.024 
& 0.88 $\!\pm\!$ 0.015 
& \textbf{5.2e-5 $\!\pm\!$ 1.4e-5} 

& 0.75 $\!\pm\!$ 0.007 
& 0.73 $\!\pm\!$ 0.009 
& \textbf{0.00 $\pm$ 0.000} \\

Spect. Dens. 
& 0.38 $\!\pm\!$ 0.013 
& 0.36 $\!\pm\!$ 0.011 
& \textbf{0.06 $\!\pm\!$ 0.003} 

& 0.25 $\!\pm\!$ 0.003 
& 0.25 $\!\pm\!$ 0.003 
& \textbf{0.05 $\pm$ 0.001} \\

Motif Dist. 
& 0.67 $\!\pm\!$ 0.014 
& 0.58 $\!\pm\!$ 0.023 
& \textbf{2.3e-4 $\!\pm\!$ 8.3e-5} 

& 0.20 $\!\pm\!$ 0.007 
& 0.18 $\!\pm\!$ 0.006 
& \textbf{3.8e-3 $\pm$ 3.8e-4} \\

Orbit Dist. 
& 0.65 $\!\pm\!$ 0.020 
& 0.54 $\!\pm\!$ 0.026 
& \textbf{4.6e-3 $\!\pm\!$ 2.4e-4} 

& 0.20 $\!\pm\!$ 0.004 
& 0.17 $\!\pm\!$ 0.004 
& \textbf{1.2e-3 $\pm$ 3.5e-5} \\

Type Degree 
& 0.28 $\!\pm\!$ 0.007 
& 0.26 $\!\pm\!$ 0.007 
& \textbf{0.05 $\!\pm\!$ 3.5e-4} 

& 0.19 $\!\pm\!$ 0.001 
& 0.19 $\!\pm\!$ 0.001 
& \textbf{0.04 $\pm$ 0.001} \\

Feat. EMD 
& 0.96 $\!\pm\!$ 0.001 
& 0.96 $\!\pm\!$ 0.001 
& \textbf{0.95 $\!\pm\!$ 0.001} 

& 0.91 $\!\pm\!$ 0.001 
& 0.90 $\!\pm\!$ 0.001 
& \textbf{0.90 $\pm$ 0.002} \\

\bottomrule
\toprule
\end{tabular}
\begin{tabular}{l c c c c c c}
 & \multicolumn{3}{c}{Country, Language and Movie} & \multicolumn{3}{c}{Country, Language, Movie and Year} \\
 \cmidrule(lr){2-4}
\cmidrule(lr){5-7}
Metrics & VGAE & VGAE-H & \modelname & VGAE & VGAE-H & \modelname\\
\midrule
Degree Dist. 
& 0.30 $\!\pm\!$ 0.002 
& 0.30 $\!\pm\!$ 0.003 
& \textbf{0.15 $\!\pm\!$ 0.001} 

& 0.28 $\!\pm\!$ 0.002 
& 0.28 $\!\pm\!$ 0.002 
& \textbf{0.14 $\pm$ 3.6e-4} \\

Clust. Coeff. 
& --- $\!\pm\!$ --- 
& 0.75 $\!\pm\!$ 0.012 
& \textbf{7.8e-6 $\!\pm\!$ 4.6e-6} 

& 0.67 $\!\pm\!$ 0.008 
& 0.66 $\!\pm\!$ 0.007 
& \textbf{0.00 $\pm$ 0.000} \\

Spect. Dens. 
& 0.23 $\!\pm\!$ 0.003 
& 0.24 $\!\pm\!$ 0.004 
& \textbf{0.04 $\!\pm\!$ 0.002} 

& 0.22 $\!\pm\!$ 0.002 
& 0.22 $\!\pm\!$ 0.002 
& \textbf{0.04 $\pm$ 2.2e-4} \\

Motif Dist. 
& --- $\!\pm\!$ --- 
& 0.15 $\!\pm\!$ 0.008 
& \textbf{0.01 $\!\pm\!$ 0.001} 

& 0.14 $\!\pm\!$ 0.003 
& 0.12 $\!\pm\!$ 0.006 
& \textbf{1.4e-3 $\pm$ 1.7e-4} \\

Orbit Dist. 
& --- $\!\pm\!$ --- 
& 0.14 $\!\pm\!$ 0.007 
& \textbf{4.5e-3 $\!\pm\!$ 3.4e-4} 

& 0.12 $\!\pm\!$ 0.002 
& 0.11 $\!\pm\!$ 0.003 
& \textbf{1.1e-3 $\pm$ 3.5e-5} \\

Type Degree 
& 0.18 $\!\pm\!$ 0.002 
& 0.17 $\!\pm\!$ 0.001 
& \textbf{0.03 $\!\pm\!$ 2.3e-4} 

& 0.18 $\!\pm\!$ 0.001 
& 0.18 $\!\pm\!$ 4.6e-4 
& \textbf{0.03 $\pm$ 2.7e-4} \\

Feat. EMD 
& 0.92 $\!\pm\!$ 0.002 
& 0.92 $\!\pm\!$ 0.002 
& \textbf{0.90 $\!\pm\!$ 0.002} 

& 0.87 $\!\pm\!$ 4.8e-4\!
& \textbf{0.87 $\!\pm\!$ 0.001} 
& 0.88 $\pm$ 0.001 \\
\bottomrule
\end{tabular}
\end{table}

\begin{table}[t]
\centering
\setlength{\tabcolsep}{2.5pt}
\caption{DBLP results for two splits (bold shows best results).} \label{tab:DBLPmain}
\begin{tabular}{l c c c c c c}
\toprule
 & \multicolumn{3}{c}{Author and Conference} & \multicolumn{3}{c}{Author, Conference, and Pub. Type} \\
 \cmidrule(lr){2-4}
\cmidrule(lr){5-7}
Metrics & VGAE & VGAE-H & \modelname & VGAE & VGAE-H & \modelname\\
\midrule 
Degree Dist. 
& 0.41 $\!\pm\!$ 0.002 
& 0.43 $\!\pm\!$ 0.007 
& \textbf{0.16 $\!\pm\!$ 0.004} 

& 0.36 $\!\pm\!$ 0.003 
& 0.35 $\!\pm\!$ 0.006 
& \textbf{0.01 $\pm$ 0.001} \\

Clust. Coeff. 
& --- $\!\pm\!$ ---  
& 1.08 $\!\pm\!$ 0.007
& \textbf{3.6e-4 $\!\pm\!$ 4.7e-5} 

& --- $\!\pm\!$ --- 
& 0.92 $\!\pm\!$ 0.018 
& \textbf{6.0e-5 $\pm$ 1.3e-5} \\

Spect. Dens. 
& 0.48 $\!\pm\!$ 0.023 
& 0.52 $\!\pm\!$ 0.020 
& \textbf{0.07 $\!\pm\!$ 0.006} 

& 0.38 $\!\pm\!$ 0.015 
& 0.36 $\!\pm\!$ 0.006 
& \textbf{1.2e-3 $\pm$ 3.2e-4} \\

Motif Dist. 
& --- $\!\pm\!$ --- 
& 0.80 $\!\pm\!$ 0.033 
& \textbf{0.15 $\!\pm\!$ 0.011} 

& --- $\!\pm\!$ --- 
& 0.43 $\!\pm\!$ 0.027 
& \textbf{6.6e-4 $\pm$ 1.7e-4} \\

Orbit Dist. 
& --- $\!\pm\!$ --- 
& 0.80 $\!\pm\!$ 0.023 
& \textbf{0.18 $\!\pm\!$ 0.006} 

& --- $\!\pm\!$ --- 
& 0.42 $\!\pm\!$ 0.013 
& \textbf{0.01 $\pm$ 0.003} \\

Type Degree 
& 0.30 $\!\pm\!$ 0.004 
& 0.29 $\!\pm\!$ 0.016 
& \textbf{0.04 $\!\pm\!$ 0.001} 

& 0.22 $\!\pm\!$ 0.002 
& 0.19 $\!\pm\!$ 0.003 
& \textbf{1.5e-3 $\pm$ 9.3e-5} \\

Feat. EMD 
& 3.44 $\!\pm\!$ 0.010
& 3.58 $\!\pm\!$ 0.011 
& \textbf{3.44 $\!\pm\!$ 0.007} 

& 3.49 $\!\pm\!$ 0.011 
& 3.51 $\!\pm\!$ 0.016 
& \textbf{3.47 $\pm$ 0.007} \\
\bottomrule
\end{tabular}
\end{table}


\begin{table}[t]
\centering
\setlength{\tabcolsep}{2.5pt}
\caption{Ablation : IMDB results for four splits} \label{tab:ablateIMDB}
\begin{tabular}{l c c c c c c}
\toprule
 & \multicolumn{3}{c}{Year} & \multicolumn{3}{c}{Country and Year} \\
 \cmidrule(lr){2-4}
\cmidrule(lr){5-7}
Metrics & NoMP & NoPool & \modelname & NoMP & NoPool & \modelname\\
\midrule

Feat. EMD 
& 0.96 $\!\pm\!$ 0.001 
& 1.00 $\!\pm\!$ 0.000 
& \textbf{0.95 $\!\pm\!$ 0.001} 

& \textbf{0.89 $\!\pm\!$ 0.001} 
& 1.00 $\!\pm\!$ 3.0e-4 
& 0.90 $\pm$ 0.002\\

\bottomrule
\toprule
\end{tabular}
\begin{tabular}{l c c c c c c}
 & \multicolumn{3}{c}{Country, Language and Movie} & \multicolumn{3}{c}{Country, Language, Movie and Year} \\
 \cmidrule(lr){2-4}
\cmidrule(lr){5-7}
Metrics & NoMP & NoPool & \modelname & NoMP & NoPool & \modelname\\
\midrule

Feat. EMD 
& 0.91 $\!\pm\!$ 0.001 
& 1.00 $\!\pm\!$ 4.2e-4  
& \textbf{0.90 $\!\pm\!$ 1.6e-3} 

& 0.88 $\!\pm\!$ 0.001 
& 1.00 $\!\pm\!$ 1.2e-4  
& \textbf{0.88 $\pm$ 0.001} \\
\bottomrule
\end{tabular}
\end{table}

\begin{table}[t]
\centering
\setlength{\tabcolsep}{2.5pt}
\caption{Ablation : DBLP results for two splits} \label{tab:ablateDBLP}
\begin{tabular}{l c c c c c c}
\toprule
 & \multicolumn{3}{c}{Author and Conference} & \multicolumn{3}{c}{Author, Conference, and Pub. Type} \\
 \cmidrule(lr){2-4}
\cmidrule(lr){5-7}
Metrics & NoMP & NoPool & \modelname & NoMP & NoPool & \modelname\\
\midrule

Feat. EMD 
& 3.45 $\!\pm\!$ 0.013
& 3.67 $\!\pm\!$ 0.010  
& \textbf{3.44 $\!\pm\!$ 0.007} 

& 3.48 $\!\pm\!$ 0.009
& 3.63 $\!\pm\!$ 0.005  
& \textbf{3.47 $\pm$ 0.007} \\
\bottomrule
\end{tabular}
\end{table}

We also construct a hybrid of VGAE and our \modelname, which we call VGAE-H.
In VGAE-H, we use VGAE in our first phase, instead of DiGress, and then use our model's second phase as a feature vector assignment module.

\textbf{Results}: Our results are presented in Table~\ref{tab:IMDBmain} for IMDB graphs and Table~\ref{tab:DBLPmain} for DBLP graphs. The results show that our approach consistently outperforms the baselines, except for one case where the hybrid approach does better. We note that some results could not be produced ($-$ in the table) for the size of graphs that we have in our experiments, given a 24-hour cutoff on the CPU and GPU. In particular, our performance is much better on the type degree distribution, showing a better modelling of the distribution of node types and connecting edges in the graphs generated by our approach.

\textbf{Ablation Study}:
In our ablation studies, we look at two components of our second phase, namely, the MP module in Phase 2, part 1 and the node pools in Phase 2, part 2 (see Figure~\ref{fig:generation}). 
In our first ablation study, we remove the MP module in Phase 2, part 1 and re-run the experiments. As a result of this removal, the sampling modules now directly interact with the skeleton graphs and the node type vectors. We call this as NoMP in the result tables.
In our second ablation study, we remove the node feature pools and replace them with dense neural networks that directly outputs the node features. We call this as NoPool in the result tables. The ablation results are shown in Tables~\ref{tab:ablateIMDB} and~\ref{tab:ablateDBLP}. We present results for Feature EMD only as the other metrics depend only on the skeleton graph output from Phase 1, which is something we do not change in our ablation set-up.
These results show that our architectural choices are important in getting the desired performance. In particular, the node feature pool is critical.

\section {Limitations and Discussions} \label{sec:limit}

An important limitation, and potential future work, is to design our approach to explicitly work with edge types. However, we do note that popular heterogeneous graph datasets in Pytorch Geometric and used in the literature do not have such explicit edge types that are not already implied by the node types. This direction is interesting future work as it also requires expanding on the labeled permutation equivariance and invariance that we have defined, though we note that our transformation of edge labeled graphs to only node labeled graphs at the start of Section~\ref{sec:methodology} may be the path towards this theoretical extension.

Another potential refinement is to have iterative refinement of feature vector assignments. Vectors once assigned are not revisited in our current approach. Our assignment is one shot in the sense that all of them are done in parallel at one go. Some options include to perform iterative refinements of our assignments, i.e., we still assign all of them at once, but then iteratively improve the assignment by updating them all again. Another option is to sequentially assign the feature vectors, that is assign feature vectors to some nodes first, and then based on such previous assignments, assign feature vectors to their neighbors and so on.

Overall, our work provides a hierarchical approach that successfully generates heterogeneous graphs with promising results for the domains that we experimented with. A number of potential extensions of our work provides ground for further research on this topic.


\bibliography{bibliography}

\begin{thebibliography}{51}
\providecommand{\natexlab}[1]{#1}
\providecommand{\url}[1]{\texttt{#1}}
\expandafter\ifx\csname urlstyle\endcsname\relax
  \providecommand{\doi}[1]{doi: #1}\else
  \providecommand{\doi}{doi: \begingroup \urlstyle{rm}\Url}\fi

\bibitem[Assouel et~al.(2018)Assouel, Ahmed, Segler, Saffari, and Bengio]{assouel2018defactor}
Rim Assouel, Mohamed Ahmed, Marwin~H Segler, Amir Saffari, and Yoshua Bengio.
\newblock Defactor: Differentiable edge factorization-based probabilistic graph generation.
\newblock \emph{arXiv preprint arXiv:1811.09766}, 2018.

\bibitem[Bojchevski et~al.(2018)Bojchevski, Shchur, Z{\"u}gner, and G{\"u}nnemann]{bojchevski2018netgan}
Aleksandar Bojchevski, Oleksandr Shchur, Daniel Z{\"u}gner, and Stephan G{\"u}nnemann.
\newblock Netgan: Generating graphs via random walks.
\newblock In \emph{International conference on machine learning}, pages 610--619. PMLR, 2018.

\bibitem[Bresson and Laurent(2019)]{bresson2019two}
Xavier Bresson and Thomas Laurent.
\newblock A two-step graph convolutional decoder for molecule generation.
\newblock \emph{arXiv preprint arXiv:1906.03412}, 2019.

\bibitem[Brody et~al.(2021)Brody, Alon, and Yahav]{brody2021attentive}
Shaked Brody, Uri Alon, and Eran Yahav.
\newblock How attentive are graph attention networks?
\newblock In \emph{International Conference on Learning Representations}, 2021.

\bibitem[Calvo-Valverde and Mena-Arias(2020)]{calvo2020evaluation}
Luis~Alexander Calvo-Valverde and Jos{\'e}~Andr{\'e}s Mena-Arias.
\newblock Evaluation of different text representation techniques and distance metrics using knn for documents classification.
\newblock \emph{Revista Tecnolog{\'\i}a en Marcha}, 33\penalty0 (1):\penalty0 64--79, 2020.

\bibitem[Chen et~al.(2021)Chen, Ma, Zhang, Wang, He, Wang, Liu, and Ma]{chen2021graph}
Chong Chen, Weizhi Ma, Min Zhang, Zhaowei Wang, Xiuqiang He, Chenyang Wang, Yiqun Liu, and Shaoping Ma.
\newblock Graph heterogeneous multi-relational recommendation.
\newblock In \emph{Proceedings of the AAAI conference on artificial intelligence}, volume~35, pages 3958--3966, 2021.

\bibitem[Chen et~al.(2012)Chen, Liu, and Yan]{chen2012drug}
Xing Chen, Ming-Xi Liu, and Gui-Ying Yan.
\newblock Drug--target interaction prediction by random walk on the heterogeneous network.
\newblock \emph{Molecular BioSystems}, 8\penalty0 (7):\penalty0 1970--1978, 2012.

\bibitem[D'Arcy et~al.(2019)D'Arcy, Corcoran, and Preece]{d2019deep}
Laura D'Arcy, Padraig Corcoran, and Alun Preece.
\newblock Deep q-learning for directed acyclic graph generation.
\newblock \emph{arXiv preprint arXiv:1906.02280}, 2019.

\bibitem[Dong et~al.(2012)Dong, Tang, Wu, Tian, Chawla, Rao, and Cao]{dong2012link}
Yuxiao Dong, Jie Tang, Sen Wu, Jilei Tian, Nitesh~V Chawla, Jinghai Rao, and Huanhuan Cao.
\newblock Link prediction and recommendation across heterogeneous social networks.
\newblock In \emph{2012 IEEE 12th International conference on data mining}, pages 181--190. IEEE, 2012.

\bibitem[Dozat(2016)]{dozat2016incorporating}
Timothy Dozat.
\newblock Incorporating nesterov momentum into adam.
\newblock 2016.

\bibitem[Fey and Lenssen(2019)]{Fey/Lenssen/2019}
Matthias Fey and Jan~E. Lenssen.
\newblock Fast graph representation learning with {PyTorch Geometric}.
\newblock In \emph{ICLR Workshop on Representation Learning on Graphs and Manifolds}, 2019.

\bibitem[Flam-Shepherd et~al.(2020)Flam-Shepherd, Wu, and Aspuru-Guzik]{flam2020graph}
Daniel Flam-Shepherd, Tony Wu, and Alan Aspuru-Guzik.
\newblock Graph deconvolutional generation.
\newblock \emph{arXiv preprint arXiv:2002.07087}, 2020.

\bibitem[Flamary et~al.(2021)Flamary, Courty, Gramfort, Alaya, Boisbunon, Chambon, Chapel, Corenflos, Fatras, Fournier, Gautheron, Gayraud, Janati, Rakotomamonjy, Redko, Rolet, Schutz, Seguy, Sutherland, Tavenard, Tong, and Vayer]{flamary2021pot}
R{\'e}mi Flamary, Nicolas Courty, Alexandre Gramfort, Mokhtar~Z. Alaya, Aur{\'e}lie Boisbunon, Stanislas Chambon, Laetitia Chapel, Adrien Corenflos, Kilian Fatras, Nemo Fournier, L{\'e}o Gautheron, Nathalie~T.H. Gayraud, Hicham Janati, Alain Rakotomamonjy, Ievgen Redko, Antoine Rolet, Antony Schutz, Vivien Seguy, Danica~J. Sutherland, Romain Tavenard, Alexander Tong, and Titouan Vayer.
\newblock Pot: Python optimal transport.
\newblock \emph{Journal of Machine Learning Research}, 22\penalty0 (78):\penalty0 1--8, 2021.
\newblock URL \url{http://jmlr.org/papers/v22/20-451.html}.

\bibitem[Fu et~al.(2020)Fu, Zhang, Meng, and King]{fu2020magnn}
Xinyu Fu, Jiani Zhang, Ziqiao Meng, and Irwin King.
\newblock Magnn: Metapath aggregated graph neural network for heterogeneous graph embedding.
\newblock In \emph{Proceedings of the web conference 2020}, pages 2331--2341, 2020.

\bibitem[G{\'o}mez-Bombarelli et~al.(2018)G{\'o}mez-Bombarelli, Wei, Duvenaud, Hern{\'a}ndez-Lobato, S{\'a}nchez-Lengeling, Sheberla, Aguilera-Iparraguirre, Hirzel, Adams, and Aspuru-Guzik]{gomez2018automatic}
Rafael G{\'o}mez-Bombarelli, Jennifer~N Wei, David Duvenaud, Jos{\'e}~Miguel Hern{\'a}ndez-Lobato, Benjam{\'\i}n S{\'a}nchez-Lengeling, Dennis Sheberla, Jorge Aguilera-Iparraguirre, Timothy~D Hirzel, Ryan~P Adams, and Al{\'a}n Aspuru-Guzik.
\newblock Automatic chemical design using a data-driven continuous representation of molecules.
\newblock \emph{ACS central science}, 4\penalty0 (2):\penalty0 268--276, 2018.

\bibitem[Guarino et~al.(2017)Guarino, Shah, and Rivas]{guarino2017dipol}
Michael Guarino, Alexander Shah, and Pablo Rivas.
\newblock Dipol-gan: Generating molecular graphs adversarially with relational differentiable pooling.
\newblock \emph{Under review}, 2017.

\bibitem[Guo and Zhao(2022)]{guo2022systematic}
Xiaojie Guo and Liang Zhao.
\newblock A systematic survey on deep generative models for graph generation.
\newblock \emph{IEEE Transactions on Pattern Analysis and Machine Intelligence}, 45\penalty0 (5):\penalty0 5370--5390, 2022.

\bibitem[Hsieh et~al.(2006)Hsieh, Hsu, and Hsu]{hsieh2006efficient}
Shu-Ming Hsieh, Chiun-Chieh Hsu, and Li-Fu Hsu.
\newblock Efficient method to perform isomorphism testing of labeled graphs.
\newblock In \emph{Computational Science and Its Applications-ICCSA 2006: International Conference, Glasgow, UK, May 8-11, 2006, Proceedings, Part V 6}, pages 422--431. Springer, 2006.

\bibitem[Hu et~al.(2020)Hu, Dong, Wang, and Sun]{hu2020heterogeneous}
Ziniu Hu, Yuxiao Dong, Kuansan Wang, and Yizhou Sun.
\newblock Heterogeneous graph transformer.
\newblock In \emph{Proceedings of the web conference 2020}, pages 2704--2710, 2020.

\bibitem[Jang et~al.(2016)Jang, Gu, and Poole]{jang2016categorical}
Eric Jang, Shixiang Gu, and Ben Poole.
\newblock Categorical reparameterization with gumbel-softmax.
\newblock \emph{arXiv preprint arXiv:1611.01144}, 2016.

\bibitem[Jo et~al.(2024)Jo, Kim, and Hwang]{jo2023graph}
Jaehyeong Jo, Dongki Kim, and Sung~Ju Hwang.
\newblock Graph generation with diffusion mixture, 2024.

\bibitem[Kearnes et~al.(2019)Kearnes, Li, and Riley]{kearnes2019decoding}
Steven Kearnes, Li~Li, and Patrick Riley.
\newblock Decoding molecular graph embeddings with reinforcement learning.
\newblock \emph{arXiv preprint arXiv:1904.08915}, 2019.

\bibitem[Khodayar et~al.(2019)Khodayar, Wang, and Wang]{khodayar2019deep}
Mahdi Khodayar, Jianhui Wang, and Zhaoyu Wang.
\newblock Deep generative graph distribution learning for synthetic power grids.
\newblock \emph{arXiv preprint arXiv:1901.09674}, 2019.

\bibitem[Kingma and Ba(2014)]{kingma2014adam}
Diederik~P Kingma and Jimmy Ba.
\newblock Adam: A method for stochastic optimization.
\newblock \emph{arXiv preprint arXiv:1412.6980}, 2014.

\bibitem[Kipf and Welling(2016)]{kipf2016variational}
Thomas~N Kipf and Max Welling.
\newblock Variational graph auto-encoders.
\newblock \emph{arXiv preprint arXiv:1611.07308}, 2016.

\bibitem[Li et~al.(2018{\natexlab{a}})Li, Zhang, and Liu]{li2018multi}
Yibo Li, Liangren Zhang, and Zhenming Liu.
\newblock Multi-objective de novo drug design with conditional graph generative model.
\newblock \emph{Journal of cheminformatics}, 10:\penalty0 1--24, 2018{\natexlab{a}}.

\bibitem[Li and Patra(2010)]{li2010genome}
Yongjin Li and Jagdish~C Patra.
\newblock Genome-wide inferring gene--phenotype relationship by walking on the heterogeneous network.
\newblock \emph{Bioinformatics}, 26\penalty0 (9):\penalty0 1219--1224, 2010.

\bibitem[Li et~al.(2018{\natexlab{b}})Li, Vinyals, Dyer, Pascanu, and Battaglia]{li2018learning}
Yujia Li, Oriol Vinyals, Chris Dyer, Razvan Pascanu, and Peter Battaglia.
\newblock Learning deep generative models of graphs.
\newblock \emph{arXiv preprint arXiv:1803.03324}, 2018{\natexlab{b}}.

\bibitem[Lim et~al.(2020)Lim, Hwang, Moon, Kim, and Kim]{lim2020scaffold}
Jaechang Lim, Sang-Yeon Hwang, Seokhyun Moon, Seungsu Kim, and Woo~Youn Kim.
\newblock Scaffold-based molecular design with a graph generative model.
\newblock \emph{Chemical science}, 11\penalty0 (4):\penalty0 1153--1164, 2020.

\bibitem[Ling et~al.(2021)Ling, Yang, and Zhao]{ling2021deep}
Chen Ling, Carl Yang, and Liang Zhao.
\newblock Deep generation of heterogeneous networks.
\newblock In \emph{2021 IEEE international conference on data mining (ICDM)}, pages 379--388. IEEE, 2021.

\bibitem[Ling et~al.(2023)Ling, Yang, and Zhao]{ling2023motif}
Chen Ling, Carl Yang, and Liang Zhao.
\newblock Motif-guided heterogeneous graph deep generation.
\newblock \emph{Knowledge and Information Systems}, 65\penalty0 (7):\penalty0 3099--3124, 2023.

\bibitem[Liu et~al.(2019)Liu, Jiang, He, Chen, Liu, Gao, and Han]{liu2019variance}
Liyuan Liu, Haoming Jiang, Pengcheng He, Weizhu Chen, Xiaodong Liu, Jianfeng Gao, and Jiawei Han.
\newblock On the variance of the adaptive learning rate and beyond.
\newblock \emph{arXiv preprint arXiv:1908.03265}, 2019.

\bibitem[Liu et~al.(2018)Liu, Allamanis, Brockschmidt, and Gaunt]{liu2018constrained}
Qi~Liu, Miltiadis Allamanis, Marc Brockschmidt, and Alexander Gaunt.
\newblock Constrained graph variational autoencoders for molecule design.
\newblock \emph{Advances in neural information processing systems}, 31, 2018.

\bibitem[Loshchilov and Hutter(2017)]{loshchilov2017decoupled}
Ilya Loshchilov and Frank Hutter.
\newblock Decoupled weight decay regularization.
\newblock \emph{arXiv preprint arXiv:1711.05101}, 2017.

\bibitem[Niu et~al.(2020)Niu, Song, Song, Zhao, Grover, and Ermon]{niu2020permutation}
Chenhao Niu, Yang Song, Jiaming Song, Shengjia Zhao, Aditya Grover, and Stefano Ermon.
\newblock Permutation invariant graph generation via score-based generative modeling.
\newblock In \emph{International Conference on Artificial Intelligence and Statistics}, pages 4474--4484. PMLR, 2020.

\bibitem[O'Bray et~al.(2021)O'Bray, Horn, Rieck, and Borgwardt]{o2021evaluation}
Leslie O'Bray, Max Horn, Bastian Rieck, and Karsten Borgwardt.
\newblock Evaluation metrics for graph generative models: Problems, pitfalls, and practical solutions.
\newblock In \emph{International Conference on Learning Representations}, 2021.

\bibitem[Pennington et~al.(2014)Pennington, Socher, and Manning]{pennington2014glove}
Jeffrey Pennington, Richard Socher, and Christopher~D Manning.
\newblock Glove: Global vectors for word representation.
\newblock In \emph{Proceedings of the 2014 conference on empirical methods in natural language processing (EMNLP)}, pages 1532--1543, 2014.

\bibitem[Qin et~al.(2023)Qin, Vignac, and Frossard]{qin2023sparse}
Yiming Qin, Clement Vignac, and Pascal Frossard.
\newblock Sparse training of discrete diffusion models for graph generation.
\newblock \emph{arXiv preprint arXiv:2311.02142}, 2023.

\bibitem[Santos et~al.(2018)Santos, Piwowarski, Denoyer, and Gallinari]{santos2018representation}
Ludovic~Dos Santos, Benjamin Piwowarski, Ludovic Denoyer, and Patrick Gallinari.
\newblock Representation learning for classification in heterogeneous graphs with application to social networks.
\newblock \emph{ACM Transactions on Knowledge Discovery from Data (TKDD)}, 12\penalty0 (5):\penalty0 1--33, 2018.

\bibitem[Simonovsky and Komodakis(2018)]{simonovsky2018graphvae}
Martin Simonovsky and Nikos Komodakis.
\newblock Graphvae: Towards generation of small graphs using variational autoencoders.
\newblock In \emph{Artificial Neural Networks and Machine Learning--ICANN 2018: 27th International Conference on Artificial Neural Networks, Rhodes, Greece, October 4-7, 2018, Proceedings, Part I 27}, pages 412--422. Springer, 2018.

\bibitem[Singh and Singh(2021)]{singh2021text}
Ritika Singh and Satwinder Singh.
\newblock Text similarity measures in news articles by vector space model using nlp.
\newblock \emph{Journal of The Institution of Engineers (India): Series B}, 102:\penalty0 329--338, 2021.

\bibitem[Su et~al.(2019)Su, Hajimirsadeghi, and Mori]{su2019graph}
Shih-Yang Su, Hossein Hajimirsadeghi, and Greg Mori.
\newblock Graph generation with variational recurrent neural network.
\newblock \emph{arXiv preprint arXiv:1910.01743}, 2019.

\bibitem[Vignac et~al.(2022)Vignac, Krawczuk, Siraudin, Wang, Cevher, and Frossard]{vignac2022digress}
Clement Vignac, Igor Krawczuk, Antoine Siraudin, Bohan Wang, Volkan Cevher, and Pascal Frossard.
\newblock Digress: Discrete denoising diffusion for graph generation.
\newblock \emph{arXiv preprint arXiv:2209.14734}, 2022.

\bibitem[Wang et~al.(2014)Wang, Yang, Zhang, and Li]{wang2014drug}
Wenhui Wang, Sen Yang, Xiang Zhang, and Jing Li.
\newblock Drug repositioning by integrating target information through a heterogeneous network model.
\newblock \emph{Bioinformatics}, 30\penalty0 (20):\penalty0 2923--2930, 2014.

\bibitem[Wasserman and Pattison(1996)]{wasserman1996logit}
Stanley Wasserman and Philippa Pattison.
\newblock Logit models and logistic regressions for social networks: I. an introduction to markov graphs and p.
\newblock \emph{Psychometrika}, 61\penalty0 (3):\penalty0 401--425, 1996.

\bibitem[You et~al.(2018{\natexlab{a}})You, Liu, Ying, Pande, and Leskovec]{you2018graph}
Jiaxuan You, Bowen Liu, Zhitao Ying, Vijay Pande, and Jure Leskovec.
\newblock Graph convolutional policy network for goal-directed molecular graph generation.
\newblock \emph{Advances in neural information processing systems}, 31, 2018{\natexlab{a}}.

\bibitem[You et~al.(2018{\natexlab{b}})You, Ying, Ren, Hamilton, and Leskovec]{you2018graphrnn}
Jiaxuan You, Rex Ying, Xiang Ren, William Hamilton, and Jure Leskovec.
\newblock Graphrnn: Generating realistic graphs with deep auto-regressive models.
\newblock In \emph{International conference on machine learning}, pages 5708--5717. PMLR, 2018{\natexlab{b}}.

\bibitem[Zhang et~al.(2019)Zhang, Jiang, Cui, Garnett, and Chen]{zhang2019d}
Muhan Zhang, Shali Jiang, Zhicheng Cui, Roman Garnett, and Yixin Chen.
\newblock D-vae: A variational autoencoder for directed acyclic graphs.
\newblock \emph{Advances in Neural Information Processing Systems}, 32, 2019.

\bibitem[Zhang et~al.(2010)Zhang, Jin, and Zhou]{zhang2010understanding}
Yin Zhang, Rong Jin, and Zhi-Hua Zhou.
\newblock Understanding bag-of-words model: a statistical framework.
\newblock \emph{International journal of machine learning and cybernetics}, 1:\penalty0 43--52, 2010.

\bibitem[Zhou et~al.(2007)Zhou, Orshanskiy, Zha, and Giles]{zhou2007co}
Ding Zhou, Sergey~A Orshanskiy, Hongyuan Zha, and C~Lee Giles.
\newblock Co-ranking authors and documents in a heterogeneous network.
\newblock In \emph{Seventh IEEE international conference on data mining (ICDM 2007)}, pages 739--744. IEEE, 2007.

\bibitem[Zhu et~al.(2022)Zhu, Du, Wang, Xu, Zhang, Liu, and Wu]{zhu2022survey}
Yanqiao Zhu, Yuanqi Du, Yinkai Wang, Yichen Xu, Jieyu Zhang, Qiang Liu, and Shu Wu.
\newblock A survey on deep graph generation: Methods and applications.
\newblock In \emph{Learning on Graphs Conference}, pages 47--1. PMLR, 2022.

\end{thebibliography}

\appendix
\newpage

\section{Data Description}

The datasets used in this work is collectively sourced and modified from \cite{fu2020magnn}, and \cite{Fey/Lenssen/2019}. The former is one of the first attempts that made use of the datasets and has the raw data available with them, while the latter is a PyTorch Geometric framework that is an indispensable tool in today's graph machine learning scenario. Specifically, we look at the Digital Bibliography \& Library Project (DBLP) dataset and the Internet Movie Database (IMDB) dataset. The original DBLP and IMDB datasets can be accessed \href{https://pytorch-geometric.readthedocs.io/en/latest/generated/torch_geometric.datasets.DBLP.html#torch_geometric.datasets.DBLP}{here} and \href{https://pytorch-geometric.readthedocs.io/en/latest/generated/torch_geometric.datasets.IMDB.html#torch_geometric.datasets.IMDB}{here}, respectively.

The original datasets, though quite convenient to use, contain only a single large heterogeneous graph, and the features available for segmenting these graphs are limited. To address this limitation, we obtained raw data from~\citep{fu2020magnn} and \href{https://dblp.org/xml/release/dblp-2024-02-01.xml.gz}{snapshot release of DBLP February 2024} to enrich the dataset with additional categorical features. This allowed us to divide the large, singular heterogeneous graph into multiple smaller ones, providing us with a more diverse and comprehensive dataset for our experiments.

\subsection{Digital Bibliography \& Library Project (DBLP)}

The original DBLP dataset is a heterogeneous graph consisting of authors (4,057 nodes), papers (14,328 nodes), terms (7,723 nodes), and conferences (20 nodes). The authors are divided into four research areas (database, data mining, artificial intelligence, information retrieval). Each author is described by a bag-of-words representation~\citep{zhang2010understanding} of their paper keywords. Each paper is described using a bag-of-words representation of their paper titles. Each term is described using GloVe vectors \cite{pennington2014glove}. The feature vector sizes are summarised in Table \ref{tab:DBLP}.

\begin{table}[!htbp]
  \caption{DBLP graph description}
  \label{tab:DBLP}
  \centering
  \begin{tabular}{lcc}
    \toprule
    Node type & Representation & Size of feature vector \\
    \midrule
    author & Bag-Of-Words of paper keywords & 334 \\
    paper & Bag-Of-Words of paper titles & 4231 \\
    term & GloVe vector & 50 \\
    \bottomrule
  \end{tabular}
\end{table}

We have modified the DBLP dataset to include another categorical feature which is the type of the publication. The original graph is very large and therefore, we have used some categorical features and their combinations to segregate the large graph into smaller components each corresponding to a unique category of the feature. For the DBLP dataset, we have used the author research area (indicated as author), conference, and type of paper publication and their combinations to obtain a set of smaller graphs. Furthermore, due to resource constraints, we have limited ourselves to use only those graphs that have number of nodes less than or equal to 200. Post this processing stage, the number of graphs in each set is highlighted in Table \ref{tab:DBLP_graphs}. Examples of the graphs from the split criteria - (author, conference and type) are visualized in Figure \ref{fig:DBLP_author_conference_type_viz}. In all DBLP graphs, authors are in red, papers in blue and terms are in green.

\begin{table}[!htbp]
  \caption{DBLP graph sets}
  \label{tab:DBLP_graphs}
  \centering
  \begin{tabular}{lc}
    \toprule
    Split criteria & Number of graphs with nodes \(\leq\) 200 \\
    \midrule
    No Split & 0 \\
    author & 0 \\
    conference & 1 \\
    type & 4 \\
    author and conference & 25 \\
    author and type & 8 \\
    conference and type & 29 \\
    author, conference and type & 77 \\
    \bottomrule
  \end{tabular}
\end{table}

\begin{figure}
    \centering
    \begin{subfigure}[t]{0.27\textwidth}
        \fbox{\includegraphics[height=0.75\textwidth]{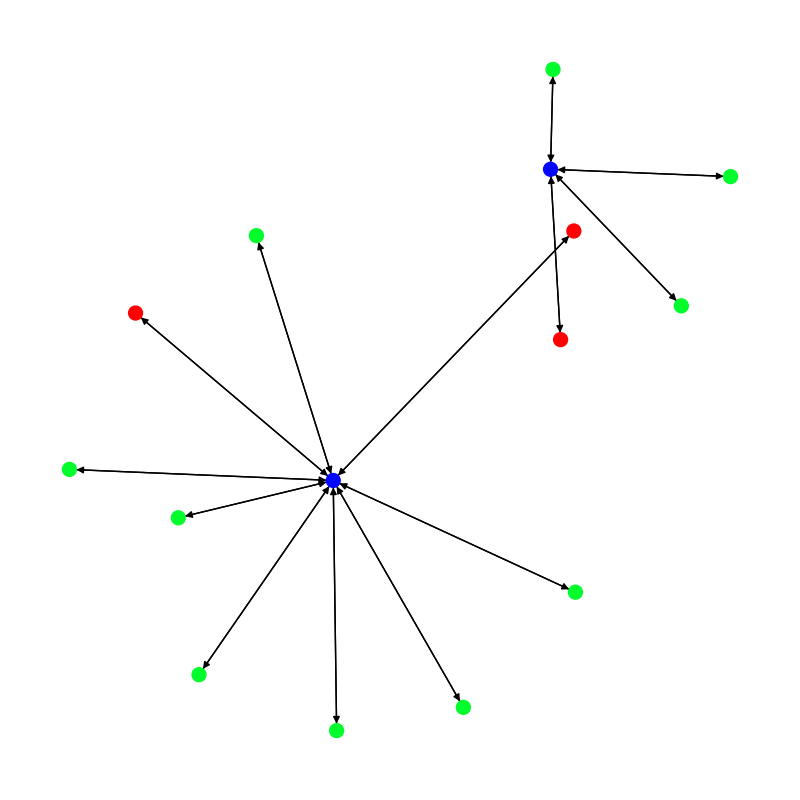}}
        
    \end{subfigure}\hspace{0.05\textwidth} 
    \begin{subfigure}[t]{0.27\textwidth}
        \fbox{\includegraphics[height=0.75\textwidth]{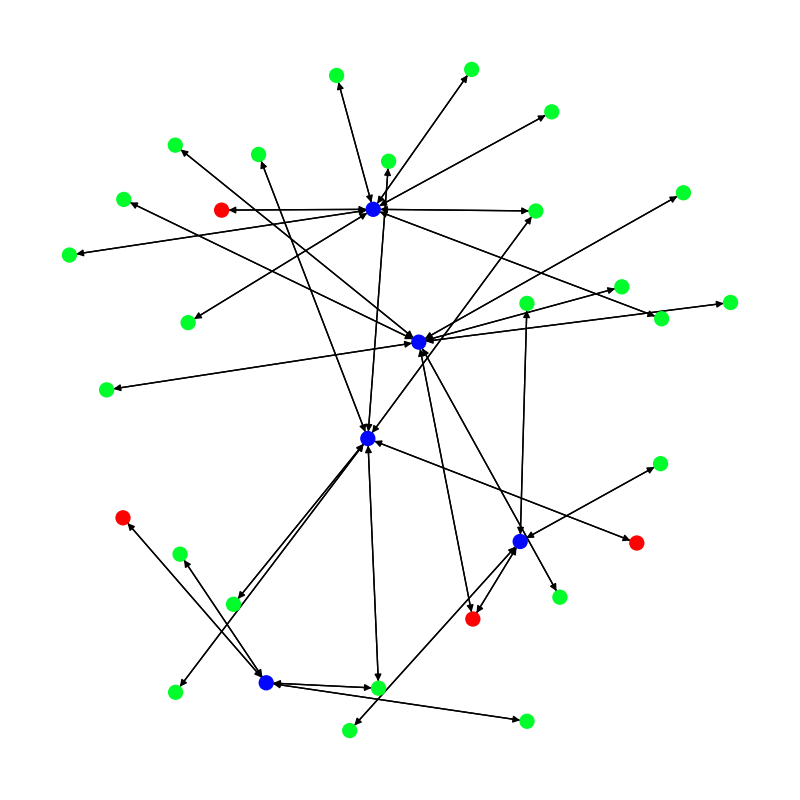}}
        
    \end{subfigure}\hspace{0.05\textwidth} 
    \begin{subfigure}[t]{0.27\textwidth}
        \fbox{\includegraphics[height=0.75\textwidth]{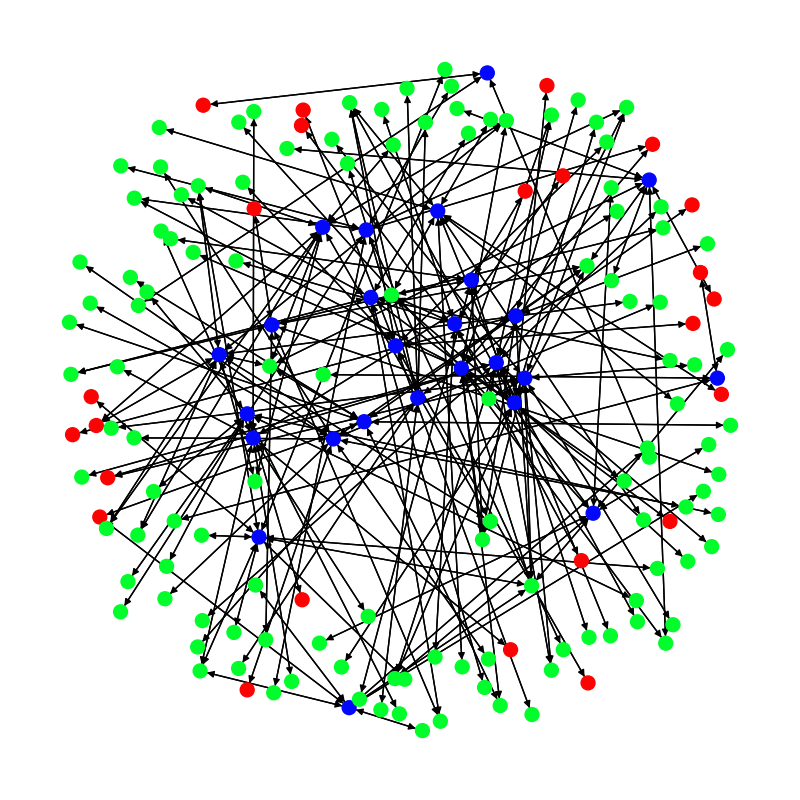}}
        
    \end{subfigure}

    \caption{Example of DBLP graphs from the split criteria - (author, conference and type)}
    \label{fig:DBLP_author_conference_type_viz}
\end{figure}

\subsection{Internet Movie Database (IMDB)}

The original IMDB dataset is a heterogeneous graph consisting of three types of entities - movies (4,278 nodes), actors (5,257 nodes), and directors (2,081 nodes). The movies are divided into three classes (action, comedy, drama) according to their genre. Movie features correspond to elements of a bag-of-words representation of its plot keywords. Features are assigned to directors and actors as the means of their associated movies' features. The feature vector sizes are summarized in Table \ref{tab:IMDB}.

\begin{table}[!htbp]
  \caption{IMDB graph description}
  \label{tab:IMDB}
  \centering
  \begin{tabular}{llc}
    \toprule
    Node type & Representation & Size of feature vector \\
    \midrule
    movie & Bag-Of-Words of plot keywords & 3066 \\
    actor & Mean of associated movies' features & 3066 \\
    director & Mean of associated movies' features & 3066 \\
    \bottomrule
  \end{tabular}
\end{table}

We have modified the dataset to include three categorical features, namely, year, language and country. For the IMDB dataset, we have used the movie classes (indicated as movie), and the three newly included categorical features and their combinations to segregate the larger graph into a set of smaller graphs. We have further restricted ourselves to only those graphs that have number of nodes less than or equal to 200, due to computational resource constraint. After the processing stage, the number of graphs in each set is highlighted in Table \ref{tab:IMDB_graphs}. Examples of the graphs from the split criteria - (movie, year, language and country) are visualized in Figure \ref{fig:IMDB_country_language_movie_year_viz}. In all IMDB graphs, movies are in red, actors in blue and directors are in green.

\begin{table}
  \caption{IMDB graph sets}
  \label{tab:IMDB_graphs}
  \centering
  \begin{tabular}{lc}
    \toprule
    Split criteria & Number of graphs with nodes \(\leq\) 200 \\
    \midrule
    No Split & 0 \\
    movie & 0 \\
    year & 64 \\
    language & 43 \\
    country & 56 \\
    movie and year & 159 \\
    movie and language & 75 \\
    movie and country & 108 \\
    year and language & 246 \\
    year and country & 441 \\
    language and country & 111 \\
    movie, year and language & 382 \\
    movie, year and country & 732 \\
    movie, language and country & 179 \\
    year, language and country & 523 \\
    movie, year, language and country & 802 \\
    \bottomrule
  \end{tabular}
\end{table}

\begin{figure}
    \centering
    \begin{subfigure}[t]{0.27\textwidth}
        \fbox{\includegraphics[height=0.75\textwidth]{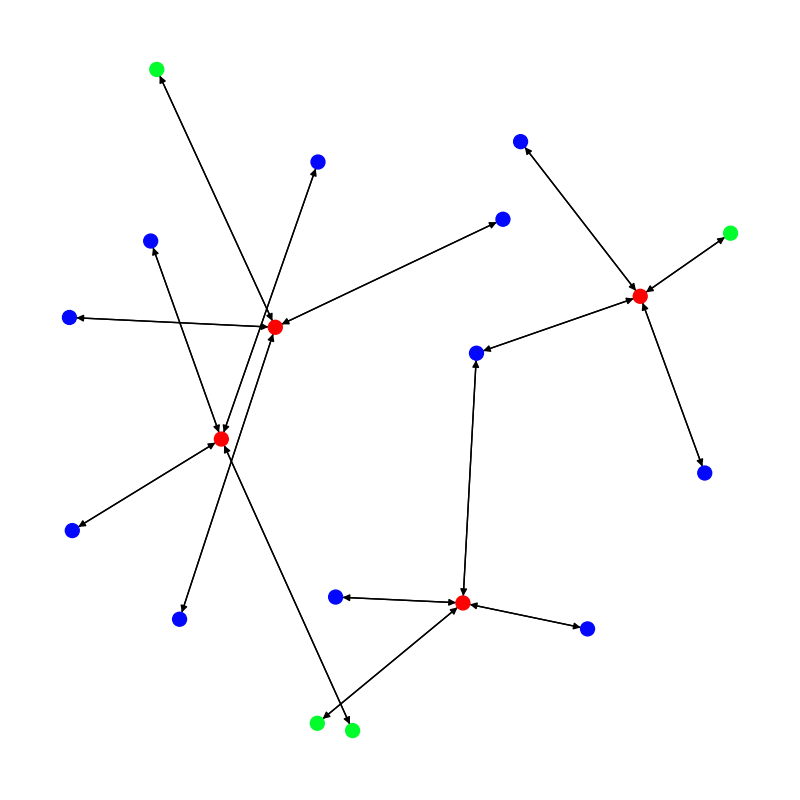}}
        
    \end{subfigure}\hspace{0.05\textwidth} 
    \begin{subfigure}[t]{0.27\textwidth}
        \fbox{\includegraphics[height=0.75\textwidth]{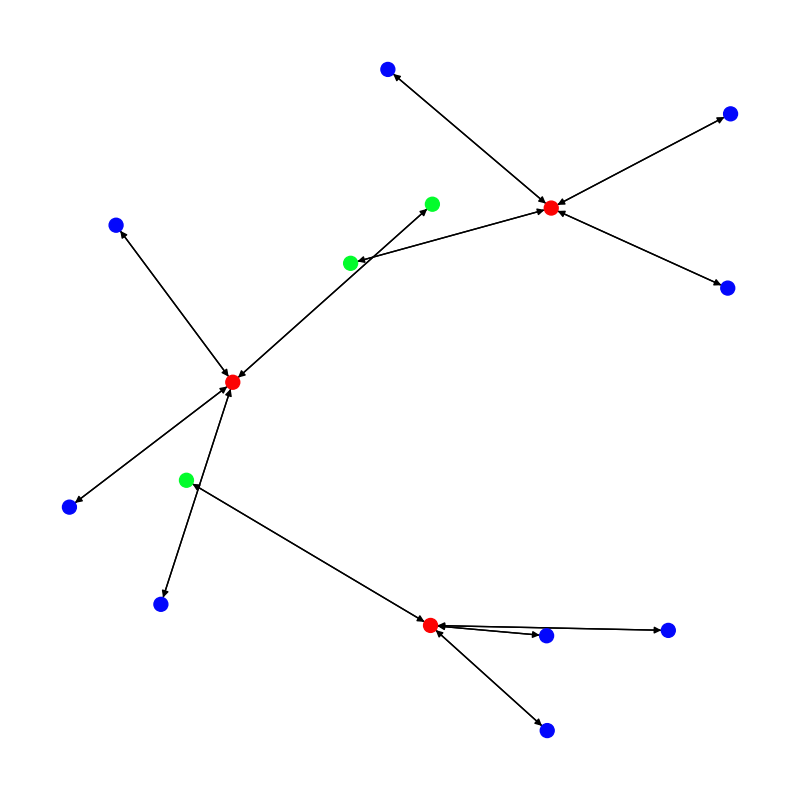}}
        
    \end{subfigure}\hspace{0.05\textwidth} 
    \begin{subfigure}[t]{0.27\textwidth}
        \fbox{\includegraphics[height=0.75\textwidth]{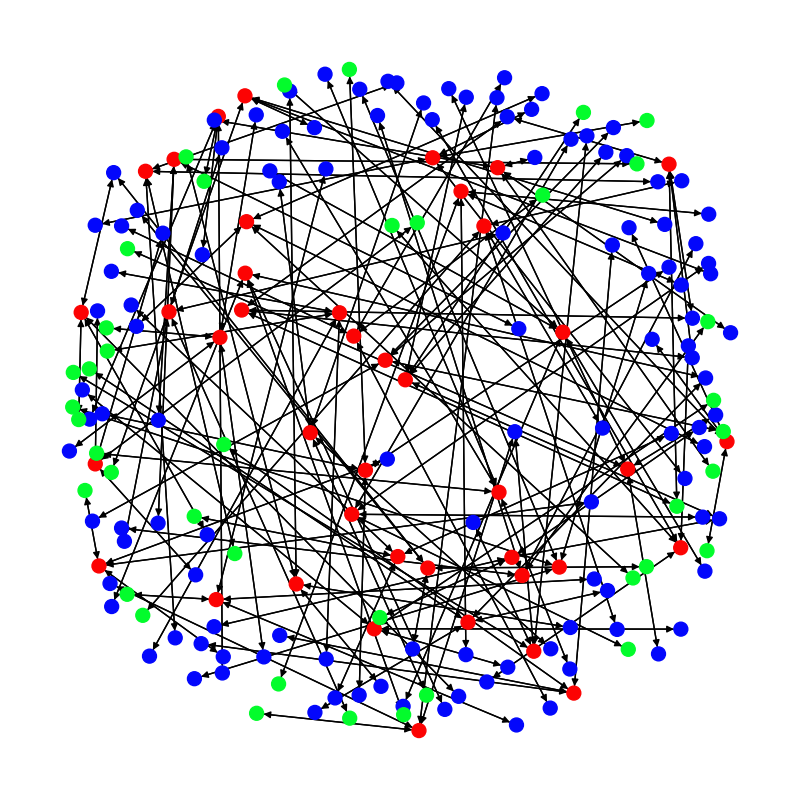}}
        
    \end{subfigure}

    \caption{Example of the IMDB graphs from the split criteria - (movie, year, language and country)}
    \label{fig:IMDB_country_language_movie_year_viz}
\end{figure}


\section{Evaluation Metrics}

In this section, the two new metrics introduced in this paper are explained in detail. 

\subsection{Feature Distribution EMD}

The distribution of feature vectors within a heterogeneous graph is a critical aspect that distinguishes it from metrics focusing solely on structural properties. This novel metric assesses the similarity between the distributions of feature vectors among generated heterogeneous graphs and real ones. In brief, the algorithm is as follows. Given a set of generated graphs \(\mathcal{P}\) and a set of real graphs \(\mathcal{Q}\), for every pair \((p, q) : p \in \mathcal{P}\) and \(q \in \mathcal{Q}\), compute a weighted average of EMD between feature vectors of each node type. This results in a cost matrix which can again be used to compute another EMD between the two whole graph sets. The details of this algorithm are given in Algorithm~\ref{alg:Feature Dist. EMD} and~\ref{alg:Graph Set EMD}, which outline the computation of this metric in pseudocode. The lower the value of the metric, the more similar the distributions are.

\SetKwComment{Comment}{/* }{ */}

\begin{algorithm}
\caption{Feature Distribution EMD between Generated and Real Graph}\label{alg:Feature Dist. EMD}
\KwData{\(p: \text{Generated Graph}, q: \text{Real Graph}\)}
\KwResult{\(emd: \text{Feature Distribution EMD}\)}
\SetKwProg{Fn}{Function}{}{end}
\Fn{FeatureDistEMD($p, q$)}{
    \(total\_nodes\_in\_q \gets 0\)\;
    \For {every node\_type $n$}{
        \(total\_nodes\_in\_q \gets total\_nodes\_in\_q + |q[n][x]|\)\;
    }

    \(total\_nodes\_in\_p \gets 0\)\;
    \For {every node\_type $n$}{
        \(total\_nodes\_in\_p \gets total\_nodes\_in\_p + |p[n][x]|\)\;
    }

    \(total\_emd1 \gets 0\)\;
    \(total\_emd2 \gets 0\)\;
    \For {every node\_type $n$}{
        $x_g \gets p[n][x]$ \Comment*[r]{Feature Matrix from $p$ of type $n$}
        $x_r \gets q[n][x]$ \Comment*[r]{Feature Matrix from $q$ of type $n$}
    
        $M \gets distance(x_g, x_r, metric)$ \Comment*[r]{$metric$ is based on node type $n$}
        a[0 ... $|x_g|$] be a new array\;
        \For {i $\leftarrow$ 0 ... $|x_g|-1$}{
            a[i] \(\gets \frac{1}{|x_g|}\)\;
        }
        
        b[0 ... $|x_r|$] be a new array\;
        \For {i $\leftarrow$ 0 ... $|x_r|-1$}{
            b[i] \(\gets \frac{1}{|x_r|}\)\;
        }
        \(emd \gets ot.emd2(a, b, M)\) \Comment*[r]{Using POT \cite{flamary2021pot} to solve the optimal transport problem and compute EMD}
        
        \(total\_emd1 \gets total\_emd1 + \frac{|q[n][x]|}{total\_nodes\_in\_q} \times emd\)
        
        \(total\_emd2 \gets total\_emd2 + \frac{|p[n][x]|}{total\_nodes\_in\_p} \times emd\)
    }

    \(total\_emd \gets (total\_emd1+total\_emd2)/2\)\;
    
    return \(total\_emd\)\;
}
\end{algorithm}

\begin{algorithm}
\caption{EMD between Generated and Real Graph Sets}\label{alg:Graph Set EMD}
\KwData{\(\mathcal{P}: \text{Set of Generated Graphs}, \mathcal{Q}: \text{Set of Real Graphs}\)}
\KwResult{\(emd: \text{Total EMD}\)}
emd \(\gets 0\)\;
\sbox0{$\vcenter{\hbox{$\begin{array}{|c|c|c|}
  \hline
  0.0 & \ldots & 0.0 \\ \hline
  \vdots & \ddots & \vdots \\ \hline
  0.0 & \ldots & 0.0 \\ \hline
\end{array}$}}$}%
$M \gets
  \underbrace{\vrule width0pt depth \dimexpr\dp0 + .3ex\relax\copy0}_{|\mathcal{Q}|}%
  \left.\kern-\nulldelimiterspace
    \vphantom{\copy0}
  \right\rbrace \scriptstyle|\mathcal{P}|
$\;

\For {r $\leftarrow$ 0 ... $|\mathcal{P}|-1$}{
    \For {c $\leftarrow$ 0 ... $|\mathcal{Q}|-1$}{
        M[r][c] $\leftarrow$ FeatureDistEMD($\mathcal{P}[r]$, $\mathcal{Q}[c]$)\;
    }
}

a[0 ... $|\mathcal{P}|$] be a new array\;
\For {i $\leftarrow$ 0 ... $|\mathcal{P}|-1$}{
    a[i] \(\gets \frac{1}{|\mathcal{P}|}\)\;
}

b[0 ... $|\mathcal{Q}|$] be a new array\;
\For {i $\leftarrow$ 0 ... $|\mathcal{Q}|-1$}{
    b[i] \(\gets \frac{1}{|\mathcal{Q}|}\)\;
}
emd \(\gets ot.emd2(a, b, M)\) \Comment*[r]{Using POT \cite{flamary2021pot} to solve the optimal transport problem and compute EMD}
\end{algorithm}

\subsection{Type Degree Distribution MMD}

This metric measures the distribution of degree of individual node types among the nodes in the graph. Computation of this metric is done as follows. Let there be $K$ number of node types in the graph $G=(V, E)$. Then, for every node $i$, compute the following $K$-length vector wherein every position $j$ in the vector is the number of neighboring nodes of $i$ that are of type $j$. Thus, they are the number of nodes $v$, such that $(i, v) \in E$ and $v$ is of type $j$. Once this vector is computed for all nodes, we have an $n \times K$ matrix for the whole graph. Drawing parallels from degree distribution MMD, note that in that metric, there was only one dimension instead of $K$, i.e., it had an $n \times 1$ dimensional matrix and we have extended it to $K$ types. Now, for each of the $K$ types, compute a degree distribution MMD separately. An average of these $K$ degree distributions is the Type Degree Distribution MMD.


\section{Hyperparameter Choices}

Here, we list the hyperparameter choices made in the first and second phase for all the datasets.

\subsection{Phase 1: Skeleton Graph Generator}
Our first phase is DiGress \cite{vignac2022digress}, and most of the model's parameters are retained to be the default ones as provided by the code of \cite{vignac2022digress}. However, the learning rate and optimizer were grid searched upon and the best model based on negative log-likelihood loss was chosen. The choices for learning rate were $[0.00001, 0.0001 0.0002, 0.001]$ and the choices for the optimizer were $[nadam, adamw]$ \cite{dozat2016incorporating} \cite{loshchilov2017decoupled}. Finally, the choices for generating the samples that were used in the second phase are summarized in Tables \ref{tab:IMDB_digress} and \ref{tab:DBLP_digress}.

\begin{table}
  \caption{Hyperparameter choices of first phase for IMDB}
  \label{tab:IMDB_digress}
  \centering
  \begin{tabular}{lcl}
    \toprule
    Split criteria & Learning Rate & Optimizer\\
    \midrule
    year & 0.0002 & nadam \\
    year and country & 0.0002 & nadam \\
    movie, language and country & 0.0002 & adamw \\
    movie, year, language and country & 0.0002 & adamw \\
    \bottomrule
  \end{tabular}
\end{table}

\begin{table}[!htbp]
  \caption{Hyperparameter choices of first phase for DBLP}
  \label{tab:DBLP_digress}
  \centering
  \begin{tabular}{lcl}
    \toprule
    Split criteria & Learning Rate & Optimizer \\
    \midrule
    author and conference & 0.0002 & nadam \\
    author, conference and type & 0.0002 & nadam \\
    \bottomrule
  \end{tabular}
\end{table}

\subsection{Phase 2: Heterogeneous Feature Vector Assignment}
For the second phase, the hyperparameter choices were extended to include the overall model's parameters along with the learning rate and optimizer. For the sake of brevity, we refer to the overall model's parameters as follows:
\begin{itemize}
    \item d - default
    \item d2 - double the parameters from default
    \item dh2 - double only the number of heads in attention based frameworks in default
\end{itemize}
Note that all the configuration files are provided in the supplementary materials with all the code. The choices for learning rate were $[0.02, 0.0002]$ and the choices for the optimizer were $[adam, radam, nadam, adamw]$ \cite{kingma2014adam} \cite{dozat2016incorporating} \cite{liu2019variance} \cite{loshchilov2017decoupled}. Finally, the choices used to generate the heterogeneous graphs are summarized in Tables \ref{tab:IMDB_fva} and \ref{tab:DBLP_fva}.

\begin{table}
  \caption{Hyperparameter choices of second phase for IMDB}
  \label{tab:IMDB_fva}
  \centering
  \begin{tabular}{lcll}
    \toprule
    Split criteria & Learning Rate & Optimizer & Overall Parameters\\
    \midrule
    year & 0.0002 & nadam & d\\
    year and country & 0.0200 & adam & d\\
    movie, language and country & 0.0002 & adam & d\\
    movie, year, language and country & 0.0200 & adam & dh2\\
    \bottomrule
  \end{tabular}
\end{table}

\begin{table}[!htbp]
  \caption{Hyperparameter choices of second phase for DBLP}
  \label{tab:DBLP_fva}
  \centering
  \begin{tabular}{lcll}
    \toprule
    Split criteria & Learning Rate & Optimizer & Overall Parameters\\
    \midrule
    author and conference & 0.0002 & nadam & d\\
    author, conference and type & 0.0200 & radam & d2\\
    \bottomrule
  \end{tabular}
\end{table}


\section{Single Graph Training}

Even though most graph generative frameworks are trained on multiple graph instances to learn the distributions of graph in the generation process as seen in \cite{guo2022systematic}, some, for instance, \citet{bojchevski2018netgan}, \cite{ling2021deep} and \citet{ling2023motif} use a different approach. In their approach to generate realistic graphs, they train only on one single graph and try to match properties of the generated graphs to this one graph. While our model uses multiple graph instances for training, however, we can also train it on a single graph (as if we have only that graph in the training set). We have extended our experiments to NetGAN~\citep{bojchevski2018netgan} which uses a single graph, and also run VGAE~\citep{kipf2016variational} on a single graph. In all datasets, the largest graph in terms of number of nodes below 200 were used for the purpose of training. The results of these experiments are summarized in Tables \ref{S_imdb} and \ref{S_dblp}.

\begin{table}[t]
\centering
\setlength{\tabcolsep}{3pt}
\caption{Single Graphs : IMDB results for four splits}
\label{S_imdb}
\begin{tabular}{l c c c c c c}
\toprule
 & \multicolumn{3}{c}{Year} & \multicolumn{3}{c}{Country and Year} \\
 \cmidrule(lr){2-4}
\cmidrule(lr){5-7}
Metrics & VGAE & NetGAN & \modelname & VGAE & NetGAN & \modelname\\
\midrule 
Degree Dist. 
& 0.37 $\!\pm\!$ 0.004 
& 0.37 $\!\pm\!$ 0.004  
& \textbf{0.23 $\!\pm\!$ 0.002}  

& \textbf{0.30 $\!\pm\!$ 0.002} 
& 0.40 $\!\pm\!$ 0.003  
& 0.38 $\!\pm\!$ 0.003\\

Clust. Coeff. 
& 0.95 $\!\pm\!$ 0.016 
& 0.73 $\!\pm\!$ 0.008  
& \textbf{6.8e-5 $\!\pm\!$ 7.2e-6} 

& 0.73 $\!\pm\!$ 0.011
& 0.77 $\!\pm\!$ 0.012  
& \textbf{6.1e-5 $\!\pm\!$ 2.6e-5} \\

Spect. Dens. 
& 0.61 $\!\pm\!$ 0.178 
& 0.64 $\!\pm\!$ 0.168  
& \textbf{0.22 $\!\pm\!$ 0.131}  

& 0.61 $\!\pm\!$ 0.081 
& 0.68 $\!\pm\!$ 0.134  
& \textbf{0.39 $\!\pm\!$ 0.053} \\

Motif. Dist. 
& 0.74 $\!\pm\!$ 0.023 
& 0.44 $\!\pm\!$ 0.008  
& \textbf{8.0e-6 $\!\pm\!$ 6.6e-7}  

& 0.25 $\!\pm\!$ 0.009 
& 0.52 $\!\pm\!$ 0.010  
& \textbf{4.2e-7 $\!\pm\!$ 1.4e-7} \\

Orbit Dist. 
& 0.63 $\!\pm\!$ 0.023 
& 0.61 $\!\pm\!$ 0.009  
& \textbf{2.7e-3 $\!\pm\!$ 7.7e-5}  

& 0.20 $\!\pm\!$ 0.006 
& 0.70 $\!\pm\!$ 0.007  
& \textbf{0.02 $\!\pm\!$ 2.5e-4} \\

Type Degree 
& 0.30 $\!\pm\!$ 0.006 
& 0.24 $\!\pm\!$ 4.6e-4  
& \textbf{0.07 $\!\pm\!$ 4.9e-4}  

& 0.19 $\!\pm\!$ 0.002 
& 0.27 $\!\pm\!$ 0.001  
& \textbf{0.13 $\!\pm\!$ 0.001} \\

Feature EMD 
& 0.97 $\!\pm\!$ 0.001 
& \textbf{0.95 $\!\pm\!$ 1.5e-4}  
& 0.98 $\!\pm\!$ 2.6e-4  

& 0.99 $\!\pm\!$ 1.9e-4 
& \textbf{0.94 $\!\pm\!$ 2.8e-4}  
& 0.98 $\!\pm\!$ 2.7e-4\\

\bottomrule
\toprule
\end{tabular}
\begin{tabular}{l c c c c c c}
 & \multicolumn{3}{c}{Country, Language and Movie} & \multicolumn{3}{c}{Country, Language, Movie and Year} \\
 \cmidrule(lr){2-4}
\cmidrule(lr){5-7}
Metrics & VGAE & NetGAN & \modelname & VGAE & NetGAN & \modelname\\
\midrule
Degree Dist. 
& 0.30 $\!\pm\!$ 0.002 
& \textbf{0.22 $\!\pm\!$ 3.6e-4}  
& 0.29 $\!\pm\!$ 0.004  

& \textbf{0.25 $\!\pm\!$ 0.002} 
& 0.60 $\!\pm\!$ 0.000  
& 0.45 $\!\pm\!$ 0.004\\

Clust. Coeff. 
& --- $\!\pm\!$ --- 
& 0.01 $\!\pm\!$ 2.2e-4  
& \textbf{5.5e-5 $\!\pm\!$ 1.5e-5}  

& 0.71 $\!\pm\!$ 0.006 
& 2.00 $\!\pm\!$ 0.000  
& \textbf{6.9e-6 $\!\pm\!$ 2.5e-6} \\

Spect. Dens. 
& 0.49 $\!\pm\!$ 0.260 
& 0.44 $\!\pm\!$ 0.068  
& \textbf{0.31 $\!\pm\!$ 0.110}  

& 0.46 $\!\pm\!$ 0.164 
& 0.75 $\!\pm\!$ 0.000  
& \textbf{0.43 $\!\pm\!$ 0.076} \\

Motif. Dist. 
& --- $\!\pm\!$ --- 
& \textbf{3.0e-5 $\!\pm\!$ 1.0e-6}  
& 3.8e-6 $\!\pm\!$ 1.1e-6  

& 0.15 $\!\pm\!$ 0.002 
& 2.00 $\!\pm\!$ 0.000  
& \textbf{0.06 $\!\pm\!$ 2.7e-5} \\

Orbit Dist. 
& --- $\!\pm\!$ --- 
& 0.06 $\!\pm\!$ 0.001  
& \textbf{0.01 $\!\pm\!$ 3.6e-4}  

& 0.12 $\!\pm\!$ 0.003 
& 2.00 $\!\pm\!$ 0.000  
& \textbf{0.02 $\!\pm\!$ 3.1e-4} \\

Type Degree 
& 0.19 $\!\pm\!$ 0.002 
& 0.22 $\!\pm\!$ 7.7e-5  
& \textbf{0.10 $\!\pm\!$ 0.002}  

& 0.17 $\!\pm\!$ 0.001 
& 0.32 $\!\pm\!$ 3.8e-4  
& \textbf{0.16 $\!\pm\!$ 0.001} \\

Feature EMD 
& 0.99 $\!\pm\!$ 0.001 
& \textbf{0.93 $\!\pm\!$ 5.6e-5}  
& 0.98 $\!\pm\!$ 1.7e-4  

& 0.99 $\!\pm\!$ 1.7e-4 
& \textbf{0.94 $\!\pm\!$ 4.2e-4}  
& 0.98 $\!\pm\!$ 4.7e-4\\
\bottomrule
\end{tabular}
\end{table}

\begin{table}[t]
\centering
\setlength{\tabcolsep}{3pt}
\caption{Single Graphs : DBLP results for two splits}
\label{S_dblp}
\begin{tabular}{l c c c c c c}
\toprule
 & \multicolumn{3}{c}{Author and Conference} & \multicolumn{3}{c}{Author, Conference, and Pub. Type} \\
 \cmidrule(lr){2-4}
\cmidrule(lr){5-7}
Metrics & VGAE & NetGAN & \modelname & VGAE & NetGAN & \modelname\\
\midrule 
Degree Dist. 
& 0.41 $\!\pm\!$ 0.003 
& 0.73 $\!\pm\!$ 0.000  
& \textbf{0.17 $\!\pm\!$ 0.001}  

& 0.41 $\!\pm\!$ 0.002 
& 0.69 $\!\pm\!$ 0.000  
& \textbf{0.30 $\!\pm\!$ 0.004} \\

Clust. Coeff. 
& --- $\!\pm\!$ --- 
& 2.00 $\!\pm\!$ 0.000  
& \textbf{0.01 $\!\pm\!$ 3.3e-4}  

& --- $\!\pm\!$ --- 
& 2.00 $\!\pm\!$ 0.000   
& \textbf{2.0e-3 $\!\pm\!$ 1.8e-4} \\

Spect. Dens. 
& 0.68 $\!\pm\!$ 0.080 
& 0.75 $\!\pm\!$ 0.000  
& \textbf{0.29 $\!\pm\!$ 0.205}  

& \textbf{0.48 $\!\pm\!$ 0.023} 
& 0.76 $\!\pm\!$ 0.000  
& 0.62 $\!\pm\!$ 0.105 \\

Motif. Dist. 
& --- $\!\pm\!$ --- 
& 2.00 $\!\pm\!$ 0.000  
& \textbf{0.33 $\!\pm\!$ 0.010}  

& --- $\!\pm\!$ --- 
& 2.00 $\!\pm\!$ 0.000  
& \textbf{0.07 $\!\pm\!$ 1.7e-4} \\

Orbit Dist. 
& --- $\!\pm\!$ --- 
& 2.00 $\!\pm\!$ 0.000  
& \textbf{0.52 $\!\pm\!$ 0.009}  

& --- $\!\pm\!$ --- 
& 2.00 $\!\pm\!$ 0.000  
& \textbf{0.41 $\!\pm\!$ 0.003} \\

Type Degree 
& 0.30 $\!\pm\!$ 0.004 
& 0.36 $\!\pm\!$ 4.1e-4  
& \textbf{0.040 $\!\pm\!$ 3.1e-4} 

& 0.30 $\!\pm\!$ 0.004 
& 0.33 $\!\pm\!$ 3.1e-4  
& \textbf{0.14 $\!\pm\!$ 0.001} \\

Feature EMD 
& \textbf{3.30 $\!\pm\!$ 0.010} 
& 3.32 $\!\pm\!$ 0.002
& 3.51 $\!\pm\!$ 0.004  

& 3.47 $\!\pm\!$ 0.007 
& \textbf{3.32 $\!\pm\!$ 0.002}  
& 3.63 $\!\pm\!$ 0.012\\
\bottomrule
\end{tabular}
\end{table}


\section{Generated Graph Samples}
We show samples of the graphs generated from our models in Figure \ref{fig:gen_graphs}. Some likeness can be seen between the example graphs from the dataset in Figure \ref{fig:DBLP_author_conference_type_viz} and \ref{fig:IMDB_country_language_movie_year_viz} and the generated graphs. For instance, red nodes (which are movies in IMDB) are centric in IMDB graphs in both the figures and in a similar vein, in DBLP graphs, blue nodes (which are papers in DBLP) are centric in both the figures.

\begin{figure}
    \centering
    \begin{subfigure}[t]{0.27\textwidth}
        \fbox{\includegraphics[height=0.75\textwidth]{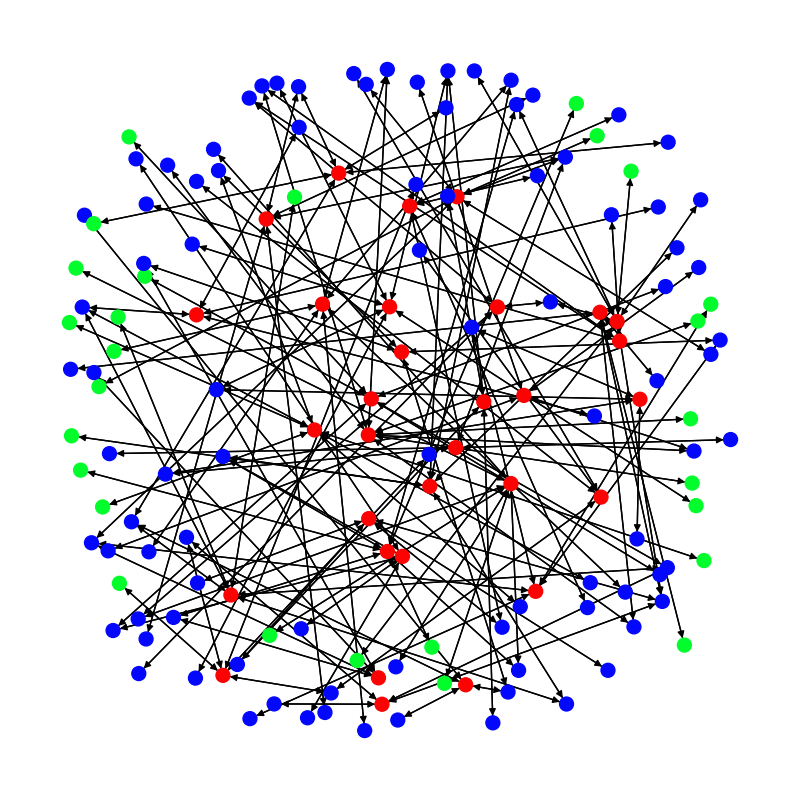}}
        \caption{IMDB from split criteria - (movie, year, language and country)}
        
    \end{subfigure}\hspace{0.05\textwidth} 
    \begin{subfigure}[t]{0.27\textwidth}
        \fbox{\includegraphics[height=0.75\textwidth]{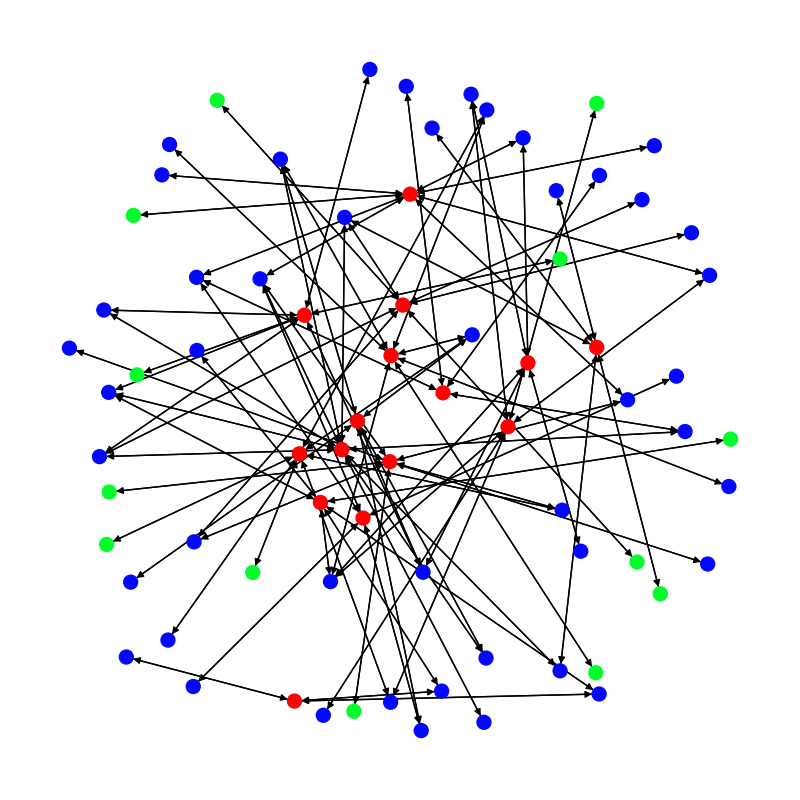}}
        \caption{IMDB from split criteria - (movie, language and country)}
        
    \end{subfigure}\hspace{0.05\textwidth} 
    \begin{subfigure}[t]{0.27\textwidth}
        \fbox{\includegraphics[height=0.75\textwidth]{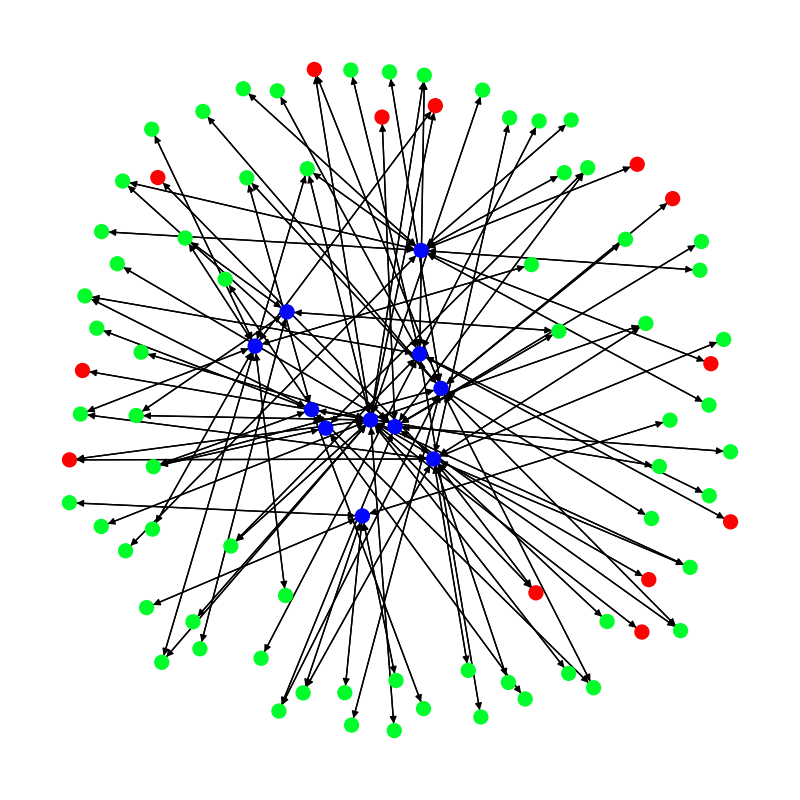}}
        \caption{DBLP from split criteria - (author, conference and type)}
        
    \end{subfigure}

    \caption{Generated Graph Samples}
    \label{fig:gen_graphs}
\end{figure}




\end{document}